\documentclass[twocolumn,journal,10pt,a4paper,twoside]{IEEEtran}
\usepackage[noadjust]{cite}
\usepackage[caption=false,font=footnotesize]{subfig}
\usepackage{psfrag}
\usepackage{graphicx,color}
\usepackage[dvipsnames]{xcolor}
\usepackage{amsmath,amsbsy,amsfonts,amssymb}
\usepackage{mathrsfs}
\usepackage{bbm}
\usepackage{comment}
\usepackage{mathtools}
\mathtoolsset{showonlyrefs}
\usepackage{Macros}
\usepackage{tikz}
\usetikzlibrary{automata,arrows,positioning,calc}

\IEEEoverridecommandlockouts

\newcommand{\am}{}

\newcommand{\revOne}{}
\newcommand{\revTwo}{}

\begin{document}

\title{Remote Monitoring of Two-State Markov Sources via Random Access Channels: an Information Freshness vs. State Estimation Entropy Perspective}
\author{Giuseppe Cocco, \IEEEmembership{Senior Member, IEEE}, Andrea Munari, \IEEEmembership{Senior Member, IEEE},\\ Gianluigi Liva, \IEEEmembership{Senior Member, IEEE}
\vspace{-1.5em}
\thanks{All authors contributed equally to this work. Part of the results in this paper have been presented at the IEEE Information Theory Workshop (ITW), April 23-28 2023, Saint-Malo, France.

G. Cocco is with the Signal Theory and Communications Department, Universitat Politecnica de Catalunya (UPC), Barcelona (Spain), and with the Department of Information and Communication Technologies, Universitat Pompeu Fabra. (email: giuseppe.cocco@upc.edu)\\ A. Munari and G. Liva are with the Institute of Communications and Navigation, German Aerospace Center (DLR), Wessling, Germany (email: \{andrea.munari, gianluigi.liva\}@dlr.de)
}
\thanks{The work of G. Cocco was supported by the Ramon y Cajal fellowship program (grant RYC2021-033908-I) funded by 
MCIN/AEI/10.13039/501100011033 and by the European Union ``NextGenerationEU'' Recovery Plan for Europe, by the Secretary of Universities and Research (Catalan Government) under a Beatriu de Pin\'{o}s fellowship, and by the European Union's Horizon 2020 research and innovation programme under the Marie Sk{\l{}}odowska-Curie grant agreement 801370.

A. Munari and G. Liva acknowledge the financial support by the Federal Ministry of Education and Research of Germany in the programme of ``Souver\"an. Digital. Vernetzt.'' Joint project 6G-RIC, project identification number: 16KISK022.
}
}

\maketitle
\thispagestyle{empty}

\markboth
    {G. Cocco et al.: Remote Monitoring of Two-State Markov Sources via Random Access Channels}
    {G. Cocco et al.: Remote Monitoring of Two-State Markov Sources via Random Access Channels}

\begin{abstract}
We study a system in which two-state Markov sources send status updates to a common receiver over a slotted ALOHA random access channel. We characterize the performance of the system in terms of state estimation entropy (SEE), which measures the uncertainty at the receiver about the sources' state. Two channel access strategies are considered: a \emph{reactive} policy that depends on the source behaviour and a \emph{random} one that is independent of it. We prove that the considered policies can be studied using two different hidden Markov models and show through a density evolution analysis that the reactive strategy outperforms the random one in terms of SEE while the opposite is true for age of information. Furthermore, we characterize the probability of error in the state estimation at the receiver, considering a maximum a posteriori and a low-complexity (decode $\&$ hold) estimator. Our study provides useful insights on the design trade-offs that emerge when different performance metrics are adopted. Moreover, we show how the source statistics significantly impact the system performance.
\end{abstract}

\section{Introduction}\label{sec:intro}

\IEEEPARstart{M}{onitoring} the state of remotely-deployed nodes in wireless sensor networks is one of the possible applications of \ac{IoT} systems. Such use cases are typically characterized by the presence of a large number of battery-powered, low-complexity devices which sense an underlying process and send updates to a common receiver over a shared channel in an often sporadic and unpredictable fashion. In these settings, grant-based solutions that require channel negotiation and reservation procedures to allocate resources prior to data delivery tend to be highly inefficient, and uncoordinated access protocols based on variations of the well-known ALOHA scheme \cite{Abramson77:PacketBroadcasting} are commonly employed to enable connectivity \cite{LoRa}. 

The main goal in remote-monitoring IoT applications is to maintain an accurate knowledge at the receiver of the status of the sensed processes. The task is in general not trivial, as it jointly depends on how nodes generate (relevant) readings, as well as on the latency experienced by packets sent in the network, and it becomes especially challenging in the presence of a distributed channel contention. Important steps towards characterizing the problem were taken with the definition of some relevant performance indicators. A pioneering role in this sense was played by the \ac{AoI}, originally introduced in the context of vehicular communications \cite{Kaul11_SECON}\cite{kaul_infoco2012}. The metric is defined as the time elapsed since the generation of the last received update for a process of interest, and captures how fresh the knowledge available at the receiver is. By virtue of its simple definition and mathematical tractability, \ac{AoI} has received a lot of research attention \cite{Yates20_Survey}, allowing to identify some fundamental trade-offs and protocol design principles that depart from those obtained using classical metrics such as throughput or latency. While initial works focused on point-to-point links, e.g. \cite{Ephremides16_TIT,Yates17_ISIT,Durisi19_JSAC,sun_poly_TIT_2020,Telatar20_TIT} among the vast available literature, important results were recently obtained also for multiple sources \cite{Yates20_TIT,Yates19_TIT,Modiano19_TNET,Pappas19,Ephremides19_Infocom,Ephremides20_CSMA},  providing fundamental insights on the behavior under ALOHA-based contention \cite{Yates17:AoI_SA,Yates20_ISIT,Uysal21_AlohaThresh,Shirin22_TIT}.

By definition, AoI is oblivious of the content of the status updates being sent, focusing only on the time of their generation, and defines a penalty that continues to grow in the absence of new incoming messages even if the monitored process does not change its state. As such, the metric  may fall short in accurately capturing the uncertainty experienced at the receiver, especially in the presence of non-memoryless sources. To overcome this limitation, alternative performance indicators have recently been proposed, such as age of incorrect information \cite{Ephremides20_TNET} or value of information \cite{Soleymani20_valueInfo}. In the first case, a (possibly non-linear) penalty is undergone only if the estimate available at the receiver is not sufficiently precise (e.g., it differs from the actual state of the source). Similarly, value of information aims at measuring the relevance of received updates towards improving the estimate of a process for a specific task. In this context, a metric of particular interest is the \ac{SEE} \cite{Rezaeian:TAC}, which quantifies the uncertainty in the knowledge of the sources' state at the sink, based on current and past channel outputs as well as on the source model. From this standpoint, while recent results were obtained in scheduled multi-user setups \cite{Liew22_TIT,Liew22_arXiv}, the behavior in random access systems is still largely unexplored. 

In this paper, we provide a contribution in this direction by detailing and extending the initial analysis presented in \cite{Cocco23_ITW}. In particular, we consider a system in which nodes monitoring two-state Markov sources communicate towards a common receiver over a slotted ALOHA channel without feedback. Under the assumption of destructive collisions, we investigate the ability of the system to acquire accurate estimates of the state of the sources at the receiver, focusing on the SEE metric. 
We consider two variations of the ALOHA access: a random transmission strategy, where the nodes send updates of their state with a fixed probability in each slot according to a Bernoulli process, and a reactive transmission strategy, where an update is sent only when a change of state in the underlying Markov source is detected. \revOne{For the first one, we show that the sampling parameters that minimize AoI, corresponding to throughput maximization \cite{kaul_infoco2012,Yates19_TIT}, also minimize SEE. The result is in line with what observed in \cite{sun_poly_TIT_2020}, and confirms that, under a random sampling policy, AoI is a good proxy for system design. On the other hand, our study reveals that a reactive approach can drastically reduce the \ac{SEE}, highlighting the importance of an access strategy that is tuned to the process being monitored. The underlying intuition follows the observation that, when the value of the monitored process is relevant, reporting only state changes makes transmissions more informative and helps reducing congestion, favoring delivery of informative updates to the receiver.}

The key contributions of the present work can be summarized as follows:
\begin{itemize}
    \item We provide an analytical characterization of the SEE for both the random and the reactive transmission strategies. In particular, we show that these approaches can be modeled using two distinct hidden Markov models. Moreover, we provide an efficient evaluation of the \ac{SEE} via a density evolution analysis \cite{RU01a}\cite[Chapter 4]{Pfister2003}, with a complexity that grows only quadratically in the number of nodes.
    \item Leaning on this, we study the behavior of the access schemes in the case of both symmetric and asymmetric sources (i.e., with different transition probabilities between the two available states). In the latter case, we propose an approximated model to simplify the DE analysis, and show by means of simulations that it provides a tight match in terms of the SEE. Comparing the trends obtained for SEE and average AoI, the trade-offs induced by a reactive transmission approach are thoroughly discussed. From this standpoint, our work provides some useful hints for protocol design in IoT monitoring systems. 
    \item \revOne{For both transmission strategies, we also study in Appendix \ref{sec:dh} the state estimation error probability achieved by a MAP estimator and by a simpler solution, dubbed decode and hold (D\&H). This strategy only updates the state estimate upon successfully receiving a message informing the receiver about the state of the tracked source. Our analysis shows that D\&H offers performance comparable to that of a MAP estimator when symmetric sources are monitored, whereas it exhibits a significant gap in the asymmetric case. }
\end{itemize}

\subsection{Related Works}
The monitoring of one or more sources through an unreliable channel has received relevant research attention in recent years. 

In the context of point-to-point channels, several works approached the problem of estimating the state of a single source. Among them, \cite{sun_poly_TIT_2020} studies the optimal sampling strategy that minimizes the mean square estimation error of a Wiener process under a sampling rate constraint, considering transmissions that incur queuing delay, and highlighting ties with an AoI-based optimization under specific conditions. In turn, \cite{gaoACC_2016} proposes an optimal transmission policy for a single sensor observing a stochastic source and transmitting the observation through a noisy channel. The considered approach takes into account past observations and decisions of the sensor, showing that a threshold-in-threshold policy is optimal under some conditions.  Joint sampling and transmission strategies for $N$-state Markov sources are tackled in \cite{Pappas23_arXiv}, taking into account performance metrics such as the estimation error probability and the cost that an estimate error might have on actuation. Scheduling policies in a battery-constrained energy harvesting monitoring system are studied in \cite{Nayyar12_Auto} for  Markov as well as Gaussian sources under different estimation distortion metrics.

Interesting results were recently derived also in multiple-access settings. In \cite{chen_liao_2022} a random access system with feedback in which each transmitter observes a different source, modelled as a random walk process, is studied. The Authors consider two possible strategies: an oblivious one \--- akin to our random transmission approach \--- and a non-oblivious one, which triggers a transmission only when the discrepancy between the current source value and the knowledge at the receiver exceeds a threshold. AoI and  a weighted sum of the squared estimation errors are employed to evaluate the policies.  In \cite{rezaeianPercom2007} a wireless sensor network is studied using the \ac{SEE} as the loss function to be minimized through optimal scheduling.  A similar setup is also considered in \cite{rezaeianArxiv2006}. Important contributions were provided lately in \cite{Liew22_TIT,Liew22_arXiv}, where the Authors investigate a system that monitors multiple binary Markov processes. Studying the uncertainty of information (UoI), i.e. the entropy of a tracked process conditioned on the latest observation, the minimization of the average sum UoI is cast onto a restless multi-armed bandit problem, deriving optimal scheduling strategies.
In \cite{ambrosinoAllerton2008} the Authors study the problem of remote estimation of a discrete-time linear dynamical source observed by multiple sensors that access a common receiver through a random access channel, tackling the trade-off between number of transmitting nodes and estimation accuracy. Finally, \cite{JACQUET2008203,luo_TIT2009} consider the problem of estimating the entropy rate of a hidden Markov model. 
\revOne{Going beyond these works, our contribution focuses on the performance at system level of different sampling strategies in terms of \ac{SEE}, \ac{AoI} and state estimation error probability. We consider two-state Markov sources, not necessarily symmetric, and compare how taking source behavior into account in the policy design impacts the considered performance metrics.}\\

The remainder of the paper is organized as follows. Sec.~\ref{sec:preliminaries} introduces the system model and the considered metrics. In Sec.~\ref{sec:optimumEstimation} we discuss the optimal state estimation approach, describing the hidden Markov models for both the random and the reactive strategies, whereas a density evolution approach to efficiently evaluate their performance is presented in Sec.~\ref{sec:de}. The fundamental trends and trade-offs of the considered strategies in terms of AoI and SEE are presented and discussed in Sec.~\ref{sec:results} by means of numerical results, followed by some concluding remarks in Sec.~\ref{sec:conclusions}.  To complement our study, we tackle in Appendix~\ref{sec:dh} the behavior in terms of state estimation error probability of a MAP estimator as well as of the simpler D\&H solution. Finally, Appendix~\ref{sec:appendix} provides details on some useful calculations.

\section{System Model and Preliminaries} \label{sec:preliminaries}

\subsection{Notation}
We use capital letters for \acp{r.v.}, and lowercase letters for their realizations. The probability of an event $\{X=x\}$ is denoted as $\prob{X=x}$, and the probability \revOne{mass} function of the \ac{r.v.} $X$ as $P(x)=\prob{X=x}$. For discrete-time, finite-state Markov chains, we denote the one-step transition probability from state $i$ to state $j$ as $q_{ij}$, and the stationary probability of state $i$ as $\pi_i$. 

\subsection{System Model}
We focus on a system with $M$ statistically independent sources (nodes) that share a common wireless channel towards a receiver. Time is slotted, and all nodes are assumed to be slot-synchronized. A source \mbox{$k\in \{0,1,\dots,M-1\}$} generates a random sequence of symbols 
$$X_0^{(k)}\, X_1^{(k)}\, X_2^{(k)}\, \ldots
$$
where $X_n^{(k)}$ belongs to the alphabet $\sourcealpha=\{0,1\}$, and represents the (random) state of the source at time (slot index) $n$. Each node is modelled as a two-state stationary Markov chain with transition probabilities
$
\tp{ij}=\mathsf{P}\left[X_n^{(k)} = j\,|\, X_{n-1}^{(k)}=i\right]
$
for all $(i,j)\in\sourcealpha\times\sourcealpha$, as reported in Fig.~\ref{fig:source}. We denote the steady-state distribution of the process as
\begin{align} 
 \pi_0 = \frac{\tp{10}}{\tp{10} + \tp{01}} \,, \quad \pi_1 = \frac{\tp{01}}{\tp{10} + \tp{01}}.
 \label{eq:statDist_MC_source}
\end{align}

\begin{figure}
\centering
    \includegraphics[width=.4\textwidth]{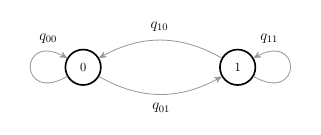}
    \vspace{-5mm}
\caption{Two-state Markov model for a generic source in the system.}
\label{fig:source}
\end{figure}

Every node in the network can transmit  update packets over the shared channel, reporting information on the state of its source to the receiver. In this work we focus on random access medium sharing policies, commonly employed in practical settings, and consider two variations of the slotted ALOHA protocol \cite{Abramson77:PacketBroadcasting} presented in detail in Sec.~\ref{sec:policies}. 
Accordingly, three possible slot outcomes can be seen at the receiver: i) \emph{idle}, i.e., no transmission is performed; ii) \emph{singleton}, i.e., only one source has transmitted; iii) \emph{collision}, i.e., two or more packets were sent concurrently. In the remainder, we assume the well-known collision channel model \cite{Abramson77:PacketBroadcasting}, so that the content of a status update is correctly received whenever sent over a singleton slot, whereas no packet can be decoded in the presence of a collision. 

We further make some additional assumptions inspired by practical systems. First, we consider that the receiver is able to detect a collision when it takes place but does not have knowledge of the number of involved packets.
Second, each transmission contains an identifier of the source, and the receiver becomes aware of the current value of a process upon decoding a packet from the corresponding source. \revTwo{Finally, no feedback is provided, so that a node is not aware of the outcome of its delivery attempt, nor can estimate the current channel load. Accordingly, no retransmission is performed.\footnote{This setup is commonly employed in many practical IoT systems, e.g. LoRaWAN \cite{LoRa}, where sensing tasks are performed by simple, battery-powered nodes that operate without feedback from the receiver.}}

Without loss of generality, we consider as reference the source with index $k=0$ and  drop the superscript in $X_n^{(k)}$. Thus, the sequence of symbols generated by the reference source will be denoted by
$$
	X_0\, X_1\, X_2\, \ldots
$$
and the receiver observes the random output sequence
$$
Y_0\,\, Y_1\,\, Y_2\,\ldots
$$
where $Y_n$ belongs to the  alphabet \mbox{$\outalpha=\{0,1,\idle,\collision,\izero,\ione\}$}. Here, $0$ and $1$ denote a collision-free observation of the corresponding state of the reference source, $\idle$ denotes an idle slot, $\collision$ denotes a collision, and $\izero,\ione$ denote a collision-free observation of the state of any of the other sources (with index in $\{1,\dots,M-1\}$), where $\izero$ represents the ``zero'' state and $\ione$ the ``one'' state.

\subsection{Transmission Strategies} \label{sec:policies}

In the remainder of our study we compare two distinct transmission strategies. In spite of their simplicity, their analysis allows to capture some fundamental trade-offs of the considered system, providing relevant insights.

\textbf{Random transmission strategy}: In the first approach, each source randomly decides at each slot whether to transmit a status update, with \emph{activation probability} $\alpha$, or to remain silent, with probability $1-\alpha$. The decision is made independently of the evolution of the source process, as well as across slots. We will refer to this approach as \emph{random transmission strategy}. In this case, the number of nodes accessing the channel over a slot follows a binomial distribution of parameters $M$ and $\alpha$. Accordingly, the probability for the reference source to deliver an update at time $n$ is given by
\begin{align}
\pup = \alpha (1-\alpha)^{M-1}
\label{eq:pup_random}
\end{align}
capturing the probability that a message is transmitted and does not undergo a collision.

\textbf{Reactive transmission strategy}: The second approach we tackle foresees a terminal accessing the channel over a slot only if a state change in the corresponding source takes place. Otherwise, no transmission is performed. For such solution, the probability for a source to deliver an update over a generic slot can be approximated as 
\begin{align}
\pup \simeq \tilde{\alpha} \,(1-\tilde{\alpha})^{M-1}
\label{eq:pup_reactive}
\end{align}
where
\begin{align}
    \tilde{\alpha} := \pi_0 \qzo + \pi_1 \qoz = \frac{2 \qzo \qoz}{\qzo+\qoz}
    \label{eq:myopic_def}
\end{align}
is the activation probability for a node. Equation \eqref{eq:pup_reactive} is exact in the symmetric case $\qzo=\qoz$, and only provides an approximation otherwise. Indeed, when $\qzo \neq \qoz$, the behavior of a node across subsequent slots is no longer i.i.d, as it depends on the current state of the source. Further details on this approximation will be discussed in Sec.~\ref{sec:myopicest}, showing that it leads to very accurate estimates of the metrics of interest. 

Some additional remarks are in order for the reactive strategy. First, in this case the access probability is fully determined by the source statistics.\footnote{To attain more flexibility, one may conceive modified strategies, introducing a probability of transmission in presence of a state change (to lower the channel load) or a probability to perform additional transmissions even in absence of state change. These modifications will not be considered in this work.} Moreover, this policy triggers a key trade-off. On the one hand, avoiding transmission of state information if no change at the source is observed may reduce the channel congestion, with beneficial effects on the overall packet delivery probability. On the other hand, the unavoidable collisions in a random access setting entail the risk of not providing updates for a long time when the source rarely changes state. Finally, it is important to observe how collisions or idle slots carry information about the state of the sources, as the access strategy intrinsically depends on whether the tracked processes experience a transition. 

\subsection{State Estimation Entropy and Error Probability}

\begin{figure}
    \centering
    \hspace{-1em}
    \subfloat[symmetric source]{
    \includegraphics[width=0.45\textwidth]{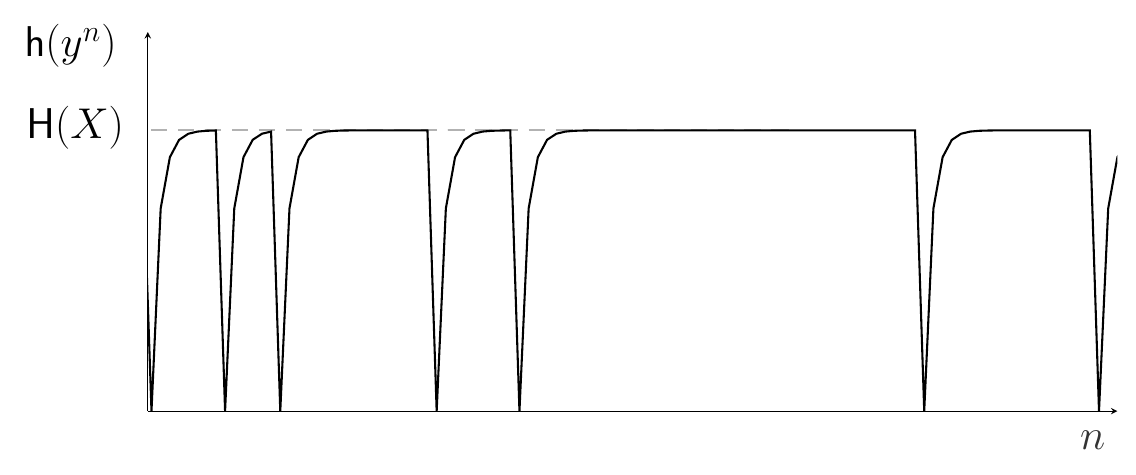}
    \label{fig:timeline_symmetric}
    }
    \vspace{1em}
    \subfloat[asymmetric source]{
    \includegraphics[width=0.45\textwidth]{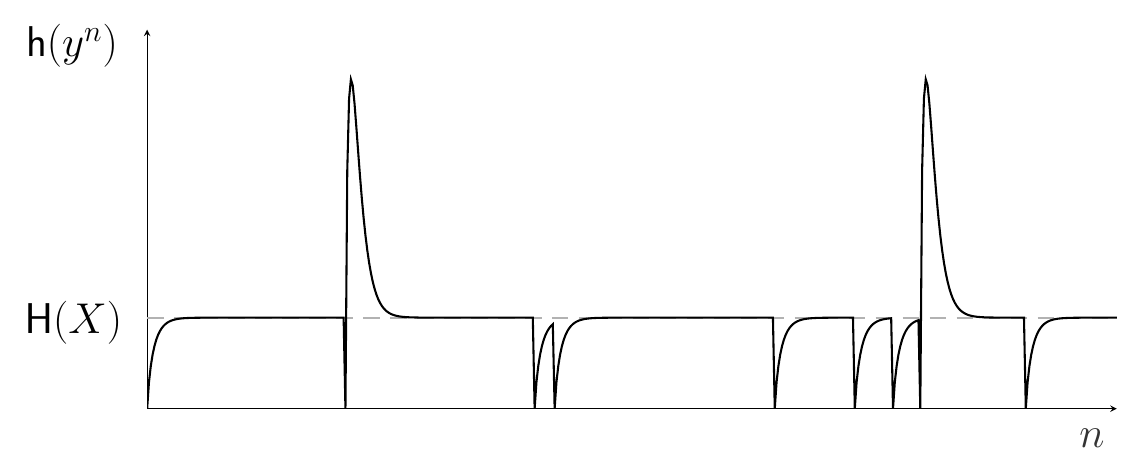}
    \label{fig:timeline_asymmetric}
    }
    \caption{Example of time evolution of $\ent(y^n)$ for a single source implementing a random transmission policy. The metric is reset to a zero whenever an update is delivered, and converges to the entropy of the source, $\ENT(X)$, when the receiver does not obtain any information for long periods of time. Subfigure (a) reports the case of a symmetric source ($\qzo=\qoz$), whereas (b) shows the behavior of an asymmetric source ($\qzo \neq \qoz$). Details are discussed in Example $1$.}
    \label{fig:timelineSEE_single}
\end{figure}

Let us denote by $Y^n$ the random vector containing the output sequence from $0$ to the current time $n$, and by $y^n$ its realization. The uncertainty experienced at the receiver about the present state of the monitored source can be conveniently captured by the entropy\footnote{The notation $\ent(\cdot)$ should not be confused with the one that is sometimes used for differential entropy.}
\begin{align}
\ent(y^n):&=\ENT(X_n|Y^n=y^n) \\
&= -\sum_{x_n\in\sourcealpha} P(x_n\given y^n) \, \log_2 P(x_n \given y^n).
\label{eq:hyn}
\end{align}

\begin{example}
To get preliminary insights on this metric, consider the case of a single source in the absence of channel contention (i.e., $M=1$, no collisions). The node follows the random transmission policy, sending (successful) updates with a certain probability at every slot. An example of the evolution over time of $\ent(y^n)$ for a symmetric source ($\qzo=\qoz$) is reported in Fig. \ref{fig:timeline_symmetric}, showing how  the uncertainty grows until an update is received, when a reset to zero denotes exact knowledge acquired at the receiver on the status of the tracked source. Note that, in the absence of refreshes, $\ent(y^n)$ approaches the entropy of the stationary distribution of the source \mbox{$\mathsf H(X) = -\pi_0 \log_2 \pi_0 - \pi_1 \log_2 \pi_1$}. The situation changes for an asymmetric source, as illustrated in Fig.~\ref{fig:timeline_asymmetric} for the case $\qoz=0.2$ and $\qzo=0.01$, corresponding to a stationary distribution $\pi_0=0.047$, $\pi_1=0.953$. Note indeed that, when an update is delivered informing that the source is in state $0$, $\ent(y^n)$ grows slowly, in view of the low probability of state transition. \revOne{For instance, the uncertainty at the receiver at the end of an idle slot following the update reception is given by $\ENT(X_n \given X_{n-1} = 0) = -\qzz \log_2 \qzz - \qzo \log_2 \qzo = 0.0808$. } Conversely, if the receiver is informed that $X$ has reached state $1$, a higher uncertainty follows in the subsequent slots, progressively reducing to converge to $\ENT(X)$ in the absence of updates. 
\revOne{In the latter example, the uncertainty after the first idle slot is ${\ENT(X_n \given X_{n-1} = 1)} = -\qoz \log_2 \qoz - \qoo \log_2 \qoo = 0.7219$, leading to the higher peaks shown in the plot.}

Two further remarks are in order. First, for the single source case, $\ent(y^n)$ is identically $0$ when a reactive transmission strategy is implemented, as the receiver can perfectly track the state of the source. Second, it is to be pointed out that the behavior of the metric becomes more involved when multiple nodes contend for the channel. In this case,  the uncertainty on the tracked source varies differently based on the outcomes observed over the slots, as well as on the implemented transmission strategy. More details will be discussed in the following sections.
\end{example}

To characterize the system performance we aim at deriving the distribution of the \ac{r.v.} 
\begin{align}
H_n := \ent(Y^n) \label{eq:Hn}
\end{align}
and, in particular, its mean value $\expect{H_n}=\ENT(X_n|Y^n)$.
More specifically, we are interested in the evaluation of the limiting behavior of such quantity as $n \rightarrow \infty$, denoted as $H_{\infty}$ and referred to as \emph{average state estimation entropy}. 
We also note that $H_{\infty}$ coincides with the time average
\begin{equation}\label{eq:Hinf1}
\lim_{N\to\infty} \frac{1}{N}\sum_{n=0}^{N-1} \ENT(X_n|Y^n).
\end{equation}
This follows by observing that ${(X_n,Y_n)}$ is a stationary stochastic process, hence $\ENT(X_n|Y^n)$ is monotonically non-increasing and  converges to a limit. The limit coincides with $H_\infty$ by the Ces\'{a}ro mean Theorem \cite[Theorem 4.2.3]{coverThomas}.

\revOne{In Appendix \ref{sec:dh}}, we also tackle the problem of estimating the reference source state at the receiver. In this context, consider a generic state estimator for $X_n$, and denote the estimate as $\est_n$. We introduce the state estimation error probability at time $n$ as
\[
P_e^{(n)} = \mathsf{P}\big[\hat{X}_n \neq X_n\big]
\]
and denote the time average of the sequence $P_e^{(n)}$ by
\begin{align}
P_e = \lim_{N\to \infty}\frac{1}{N}\sum_{n=0}^{N-1} P_e^{(n)}.
\label{eq:Pe_def}
\end{align}

\subsection{Age of Information} \label{sec:aoi_metric}
As a reference benchmark for our study, we \revOne{also analyze} the performance of the presented schemes in terms of \ac{AoI}. 
The metric, originally proposed in \cite{Kaul11_SECON}, is a well-established measure for the notion of information freshness, capturing how outdated the knowledge about the state of a source is at the destination. To introduce this quantity, we assume each status update to contain a time stamp, denoting the instant at which the message was generated. Accordingly, the current AoI for the source of interest at time $t$ is defined as
\begin{align}
    \Delta(t) := t - \sigma(t)
\end{align}
where $\sigma(t)$ is the time stamp of the last successfully received update from the node. Leaning on this definition, $\Delta(t)$ follows a sawtooth profile, growing linearly over time and being reset each time an update is received, as exemplified in Fig.~\ref{fig:timelineAoI}. For the setting under study, we assume that the  time stamp of a message corresponds to the the start of that slot it is sent over, so that the \ac{AoI} falls to one slot duration if the packet is successfully decoded, accounting for the transmission and reception time of the message over the channel. 
Via simple arguments, the stochastic process $\Delta(t)$ can be shown to be ergodic (see, e.g. \cite{Munari21_TCOM_AoI}), and we will focus in the remainder on its average value
\begin{align}
    \bar{\Delta} := \expect{\Delta(t)} = \lim_{N\rightarrow\infty}\frac{1}{N} \sum_{n=0}^{N-1} \Delta(n)
\end{align}
where $n = \lfloor t \rfloor$.
For a slotted ALOHA access, assuming independent behavior of all nodes across slots, $\bar{\Delta}$ takes the simple form \cite{Munari21_TCOM_AoI}
\begin{align}
    \bar{\Delta} = \frac{1}{2} + \frac{1}{\pup}
    \label{eq:aoi_SA}
\end{align}
where $\pup$ is given for the random and reactive transmission strategies in \eqref{eq:pup_random} and \eqref{eq:pup_reactive}, respectively. Note that, from \eqref{eq:aoi_SA}, the metric is minimized by maximizing $\pup$. In other words, for slotted ALOHA the optimal strategy coincides with throughput maximization.

\begin{figure}
    \centering
    \includegraphics[width=.45\textwidth]{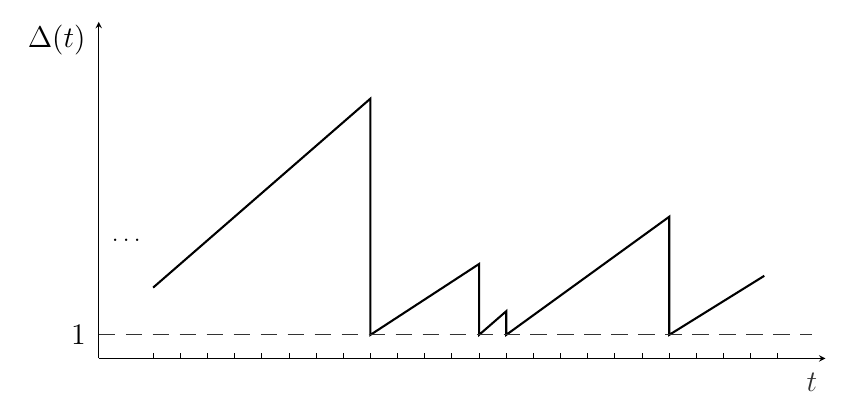}
    \vspace{-3mm}
    \caption{Example of time evolution of age of information for a source of interest at the destination. The metric is reset to a value of one slot whenever an update is received, and grows linearly otherwise.}
    \label{fig:timelineAoI}
\end{figure}
\section{\revOne{Optimum State Estimation}} \label{sec:optimumEstimation}

As a first step towards characterizing the performance of the system under study, we address the problem of estimating the state $X_n$ of the reference source at the receiver upon observing a sequence of channel outputs $y^n$. In particular, we lean on the \ac{APP} logarithmic ratio
\begin{equation}\label{eq:lambdan_def}
\lambda_n := \ln \frac{\prob{X_n=0 \given Y^n = y^n}}{\prob{X_n=1 \given Y^n = y^n}}
\end{equation}
which, combined with a threshold test ($\hat{x}_n=0$ if $\lambda_n>0$ and $\hat{x}_n=1$ otherwise), yields an optimum \ac{MAP} estimator, i.e, minimizing the state estimation error probability. 
In the remainder of this section, we introduce hidden Markov models to capture the relation between the observed channel outputs and the evolution of the reference source, and use them to develop  recursive equations to efficiently compute \eqref{eq:lambdan_def}. To this aim, we also show that the \ac{APP} is a sufficient statistics for $X_n$. The results will be used in Sec.~\ref{sec:de} to derive  \revOne{the distribution} of the r.v. $H_n$ and eventually the average state estimation entropy.

\revTwo{\begin{lemma}\label{lemma:ss}
	Assume $X_0, X_1, \ldots$ to be the random sequence of states generated by a two-state stationary Markov source, and let $Y_0, Y_1, \ldots$ be the state sequence observations. Then, $\lambda_n$ 	is a sufficient statistic for estimating $X_n$ given $Y^n$.
\end{lemma}}
\revTwo{\begin{proof}
	Leaning on the Fisher-Neyman factorization theorem, it suffices to show that $P(y^n|x_n)$ can be written as $a(x_n,\lambda_n) \,b(y^n)$, $a(\cdot)$ and $b(\cdot)$ being non-negative functions. Following the approach in \cite[Lemma 4.7]{RU08}, we observe that
	\[
	\ln \frac{P(x_n \given  y^n)}{\prob{X_n=0 \given Y^n = y^n}}=\left\{
	\begin{array}{ll}
		0 & \text{if } x_n=0\\
		-\lambda_n & \text{if } x_n=1
	\end{array}
	\right.
	\]
	which implies that
\[
	P(x_n|y^n) = \prob{X_n=0 \given Y^n = y^n} \exp(-x_n\lambda_n).
	\]
	By Bayes' rule, we then have	
    \begin{align}
	   P(y^n|x_n) 	&= \frac{P(x_n|y^n) P(y^n)}{P(x_n)}\\
					&=  \frac{\prob{X_n=0 \given Y^n = y^n} \exp(-x_n\lambda_n) P(y^n)}{P(x_n)}\\
					&=a(x_n,\lambda_n) b(y^n)
	\end{align}
where $a(x_n,\lambda_n) = \exp(-x_n\lambda_n)/P(x_n)$
and $b(y^n) = \prob{X_n=0 \given Y^n = y^n} P(y^n)$.
\end{proof}} 

\revTwo{Lemma \ref{lemma:ss} allows to characterize the distribution of the r.v. $H_n$ defined in \eqref{eq:Hn}. Note indeed that $X_n \rightarrow \Lambda_n \rightarrow Y^n$, i.e., they form a Markov chain. Leaning on this, we have
\revOne{\begin{equation}\label{eq:def_LLR}
\ln \frac{\prob{X_n=0 \given Y^n = y^n}}{\prob{X_n=1 \given Y^n = y^n}} = \ln \frac{\prob{X_n=0 \given \Lambda_n = \lambda_n}}{\prob{X_n=1 \given \Lambda_n = \lambda_n}}
\end{equation}
and hence, owing to \eqref{eq:lambdan_def}, $\prob{X_n\!=\!x_n\! \given\! \Lambda_n\! =\! \lambda_n} = \exp(-x_n \lambda_n)/(1\!+\!\exp(-\lambda_n))$.}
\revOne{We can use the obtained probabilities to express the entropy \eqref{eq:hyn} in terms of the \ac{APP} logarithmic ratio} as 
\begin{align}
\begin{split}
\ent(y^n)&=\ENT(X_n|Y^n=y^n) 
		=\ENT(X_n|\Lambda_n=\lambda_n)\\
		&=\sum_{x_n\in\sourcealpha} \frac{\exp(-x_n \lambda_n)}{1+\exp(-\lambda_n)} \log_2 \left( \frac{1+\exp(-\lambda_n)}{\exp(-x_n \lambda_n)} \right).
\end{split}\label{eq:ent_yn} 
\end{align}
With a slight abuse of notation, we denote the leftmost term of \eqref{eq:ent_yn} as $\ent(\lambda_n)$.}

\begin{figure*}[!b]
	\setcounter{equation}{10}
	\hrulefill
	\begin{align}
		\begin{split}
				\lambda_n  &= \ln \frac{\sum_{x_{n-1}\in\sourcealpha} \prob{X_n=0 , X_{n-1}=x_{n-1}, Y^{n-1} = y^{n-1}, Y_n=y_n }}{\sum_{x_{n-1}\in\sourcealpha} \prob{X_n=1 , X_{n-1}=x_{n-1}, Y^{n-1} = y^{n-1}, Y_n=y_n }}\\[.3em]		
				  &\equal{(a)} \ln \frac{\sum_{x_{n-1}\in\sourcealpha} \prob{X_n=0 , Y_n=y_n \given X_{n-1}=x_{n-1}} P(x_{n-1}| y^{n-1})}{\sum_{x_{n-1}\in\sourcealpha} \,\prob{X_n=1 , Y_n=y_n \given X_{n-1}=x_{n-1}} P(x_{n-1}| y^{n-1})}\\[.3em]
				  &\equal{(b)} \ln \frac{P(y_n\given 0)}{P(y_n\given 1)} + \ln \frac{\sum_{x_{n-1}\in\sourcealpha} \,\prob{X_n=0 \given X_{n-1}=x_{n-1}} P(x_{n-1}| \lambda_{n-1})}{\sum_{x_{n-1}\in\sourcealpha} \, \prob{X_n=1 \given X_{n-1}=x_{n-1}} P(x_{n-1}| \lambda_{n-1})}\\[.3em]
				 &\equal{(c)} \ln \frac{P(y_n\given 0)}{P(y_n\given 1)} + \ln \frac{\tp{00} + \tp{10}\exp(-\lambda_{n-1}) } {\tp{01} + \tp{11}\exp(-\lambda_{n-1})}
				  =: f(y_n,\lambda_{n-1}). 
		\end{split} \label{eq:lambda_rnd}
	\end{align}
\end{figure*}
\setcounter{equation}{9}

\revTwo{\begin{remark}
By observing that
$H_n=\ent(\Lambda_n)$ 
we see that the distribution of the \ac{r.v.} $H_n$ can be derived from the distribution of the \ac{APP} logarithmic ratio $\Lambda_n$. 
\end{remark}}

\revTwo{We now focus on deriving a recursive formulation to obtain the \ac{APP} logarithmic ratio $\lambda_n$ as a function of its previous value $\lambda_{n-1}$ and of the channel observation $y_n$, for both the random and reactive transmission strategies. The recursive formulation is based on \acp{HMM}. In particular, the statistical relation between the output sequence $Y^n$ and the reference source sequence $X^n$ can be suitably described via different \acp{HMM}, depending on the transmission strategy adopted by the nodes, as described next.}

\vspace{-1em}
\revTwo{\subsection{Random Transmission Strategy}}
\revTwo{\subsubsection{Hidden Markov Model}
In this case, the observation of channel outputs in $\{\collision,\izero,\ione\}$ can be assimilated to the observation of an idle slot, i.e., the knowledge of the state of the other sources does not provide information about the state of the reference source (we will see that this is not true for the reactive policy). From this and the memoryless nature of the access strategy, it follows that the statistical relation between the output sequence $Y^n$ and  $X^n$ is fully characterized by the conditional probability function 
\begin{align}
P(Y^n|X^n)=\prod_{\ell=0}^n P(y_\ell\given x_\ell).
\label{eq:condProb_HMM_random}
\end{align}
The distribution in \eqref{eq:condProb_HMM_random} can readily be derived leaning on the activation probability $\ptx$ of the nodes, as well as of the underlying Markov process $X^n$, as detailed in Appendix~\ref{sec:appendix}.}

\subsubsection{Recursive APP Logarithmic Ratio Calculation}
\revOne{We can rewrite} \eqref{eq:lambdan_def} as
\begin{align}
\lambda_n = \ln \frac{\prob{X_n=0 , Y^n = y^n }}{\prob{X_n=1 , Y^n = y^n}}.\label{eq:lambdan}
\end{align}
Following the well-known steps for the derivation of the forward-backward algorithm recursions over the presented \acp{HMM}  (see e.g. \cite{BCJR,RabinerHMM}), we obtain \eqref{eq:lambda_rnd} reported at the end of the page.
In the expression, (a) follows by \revOne{conditioning on $X_{n-1}$ and $Y^{n-1}$} and by observing that  $(X_n, Y_n)$ are independent on $Y^{n-1}$ once we condition on $X_{n-1}$. Similarly, (b) follows by application of  Bayes' Theorem and by observing that  $Y_n$ is independent on $X_{n-1}$ once we condition on $X_n$. Moreover, since $\lambda_{n-1}$ is a sufficient statistic for $x_{n-1}$, we can replace $P(x_{n-1}|y^{n-1})$ with $P(x_{n-1}|\lambda_{n-1})$. Finally, (c) is obtained by introducing the Markov source transition probabilities, and by noting that $P(x_{n-1}|\lambda_{n-1}) \propto \exp(-x_{n-1}\lambda_{n-1})$, as derived after Lemma $1$.
The recursion is hence defined via \mbox{$\lambda_n = f(y_n,\lambda_{n-1})$}, and its evaluation complexity is independent of the number of sources.
\setcounter{equation}{11}

\medskip

\revTwo{\subsection{Reactive Transmission Strategy}}

\revTwo{We consider next the reactive transmission strategy. As before, we first introduce the \ac{HMM}, followed by the derivation of the recursive \ac{APP} logarithmic ratio. In addition, we provide an alternative low-complexity (sub-optimal) state estimation algorithm.}

\medskip

\revTwo{\subsubsection{Hidden Markov Model}\label{sec:HMM_reactive}
The observation of channel outputs in $\{\idle,\collision,\izero,\ione\}$ should be used in this case to refine the estimate of the state of the reference source. To see why this is true, let us consider the following example.}

\revTwo{\begin{example}
Assume the case $M=3$, with sources being driven by the transition probabilities $\tp{00}=0.1$, $\tp{01}=0.9$, $\tp{10}=0.1$, $\tp{11}=0.9$, and following a reactive strategy. Note that a state change, and hence a transmission, is much more probable if the past state is $0$. Suppose now that the receiver observes at time $n-1$ the output $Y_{n-1}=0$ (the past state of the reference source is $0$), whereas a collision is experienced at slot $n$, i.e. $Y_n=\collision$. Consider the following different situations for the other two sources at time $n-1$:
 \begin{itemize}
 	\item[a)]  $X^{(1)}_{n-1}=0$ and $X^{(2)}_{n-1}=0$;
 	\item[b)]  $X^{(1)}_{n-1}=1$ and $X^{(2)}_{n-1}=1$.
 \end{itemize}
\revTwo{The probability for the reference source to have transitioned given that a collision is observed in slot $n$ can be computed using the definition of conditional probability as
\begin{align*}
    &\mathsf P[X_n=1 \given Y_n=\collision, X_{n-1}=0, X^{(1)}_{n-1}, X^{(2)}_{n-1}] \\[.3em]
	&\hspace{1em}=     \frac{\mathsf P[X_n=1, Y_n=\collision \given X_{n-1}=0, X^{(1)}_{n-1}, X^{(2)}_{n-1} ]}
    {\mathsf P[Y_n = \collision \given X_{n-1}=0, X^{(1)}_{n-1}, X^{(2)}_{n-1}]}.
\end{align*}
In case (a), all three sources have the same probability of transition to a different state at time $n$, and hence of generating a transmission in slot $n$. Accordingly, the numerator evaluates to $\qzo(1-\qzz^2)$, whereas the denominator is given by $\qzo^3 + 3\qzz \qzo^2$. It follows that the probability of a transition for the reference node is $\approx 0.91$. Conversely, in case (b), the joint probability of the reference source transitioning and of seeing a collision is given by $\qzo(1-\qoo^2)$, whereas the overall collision event has probability $\qzo\qoz^2 + 2\qzo\qoz\qoo + \qzz\qoz^2$. The transition of the reference source can then be inferred in this case with probability  $\approx 0.99$.
}
\end{example}}

\revTwo{From the example above we can see that, under reactive sampling, having (even partial) knowledge of the state of the other sources, jointly with the channel output observations, provides information on the state of the reference source. From a careful inspection of the example we also see that what matters is not the state of each source, but rather the number of sources that are in a given state during the past slot.
Accordingly, we denote by $S_n$ the \ac{r.v.} that counts the number of sources (with the exclusion of the reference one) that are in state $0$ at time $n$. Obviously, $S_n \in \countset$ with $\countset = \{0, 1,\dots, M-1\}$. For the sake of estimating the state of the reference source, the system can be  characterized as a $2M$-state Markov chain $\sigma_n:=(X_n, S_n)$, with state space $\sourcealpha \times \countset$. 
The channel output depends then on the system state transition through the conditional probability function $P(y_n\given \sigma_{n-1}, \sigma_n)$, leading to the HMM illustrated in Fig.~\ref{fig:HMM_reactive_full}.}

\begin{figure}
    \centering
    \includegraphics[height=.95\columnwidth]{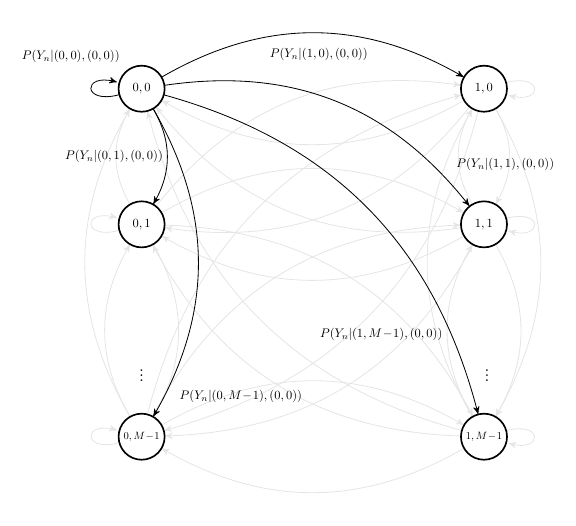}
    \vspace{-5mm}
    \caption{Hidden Markov model for the reactive transmission strategy. The underlying Markov chain has state $\sigma_n = (X_n,S_n)$, where $S_n$ denotes the number of nodes, other than the tracked source, that are in state $0$ at time $n$. The observed outputs depend on the state through the conditional probabilities $P(y_n\given \sigma_{n-1}, \sigma_n)$, some of which are highlighted in the diagram.}
    \label{fig:HMM_reactive_full}
\end{figure}

\revTwo{\begin{remark}\label{rem:symmetric_TP}
Note that, when $\tp{01}=\tp{10}$ (symmetric sources), the information on the counter $S_n$ can be dropped without any information loss, as the knowledge of $S_{n}$ does not influence the probability of observing a collision at step $n+1$. In this case, the conditional probability function $P(y_n\given x_{n-1},x_n)$ suffices, and the HMM simplifies to the one reported in Fig.~\ref{fig:hmm_reactive_symm}.
\end{remark}}

\begin{figure}
    \centering
    \includegraphics[width=.4\textwidth]{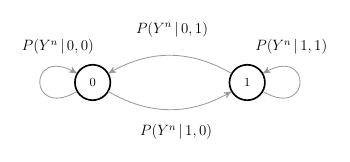}
    \caption{Hidden Markov model for the reactive transmission policy, when symmetric sources are observed ($\qzo=\qoz$).}
    \label{fig:hmm_reactive_symm}
\end{figure}

\begin{figure*}[!b]
	\setcounter{equation}{13}
	\hrulefill
	\begin{align}
		\begin{split}
			\tilde{\lambda}_n &= \ln \frac{\sum_{x_{n-1}\in\sourcealpha} \, \prob{X_n=0 , X_{n-1}=x_{n-1}, Y^{n-1} = y^{n-1}, Y_n=y_n }}{\sum_{x_{n-1}\in\sourcealpha} \, \prob{X_n=1 , X_{n-1}=x_{n-1}, Y^{n-1} = y^{n-1}, Y_n=y_n }}\\[.4em]	
			&\equal{(a)} \ln \frac{\sum_{x_{n-1}\in\sourcealpha} \, \prob{Y_n=y_n \given X_{n-1}=x_{n-1},X_n=0} \prob{X_n=0 \given X_{n-1}=x_{n-1}} P(x_{n-1}| y^{n-1})}{\sum_{x_{n-1}\in\sourcealpha} \, \prob{Y_n=y_n \given X_{n-1}=x_{n-1},X_n=1}\prob{X_n=1 \given X_{n-1}=x_{n-1}} P(x_{n-1}| y^{n-1})}\\[.4em]
			&\equal{(b)} \ln \frac{\sum_{x_{n-1}\in\sourcealpha} \, \prob{Y_n=y_n \given X_{n-1}=x_{n-1},X_n=0} \tp{x_{n-1}0} \exp(-x_{n-1}\tilde{\lambda}_{n-1})}{ \sum_{x_{n-1}\in\sourcealpha} \, \prob{Y_n=y_n \given X_{n-1}=x_{n-1},X_n=1}\tp{x_{n-1}1} \exp(-x_{n-1}\tilde\lambda_{n-1})}\\
			&=: g(y_n,\tilde{\lambda}_{n-1}). 
		\end{split}
		\label{eq:APPLRMy}
	\end{align}
\end{figure*}
\setcounter{equation}{11}

\revTwo{\subsubsection{Recursive APP Logarithmic Ratio Calculation}
In the reactive case, the derivation of the \ac{APP} logarithmic ratio is based on the recursive computation of the state probabilities for the $2M$-states \ac{HMM} described in Sec.~\ref{sec:HMM_reactive}. Following the steps of the forward-backward algorithm recursions, \revTwo{and recalling that $\sigma_n=(X_n, S_n)$} we have:
\begin{align}
	\begin{split}
	&P(\sigma_n,y^n) \equal{(a)} \sum_{\sigma_{n-1} \in \sourcealpha\times\countset} P(\sigma_n,\sigma_{n-1},y_n,y^{n-1})\\
	&\equal{(b)} \sum_{\sigma_{n-1} \in \sourcealpha\times\countset} P(\sigma_n,y_n|\sigma_{n-1},y^{n-1})P(\sigma_{n-1},y^{n-1})\\
	&\equal{(c)} \sum_{\sigma_{n-1} \in \sourcealpha\times\countset} P(y_n|\sigma_{n-1},\sigma_n)P(\sigma_n|\sigma_{n-1})P(\sigma_{n-1},y^{n-1}).
	\end{split}
	\label{eq:P_react}
\end{align}
In \eqref{eq:P_react}, (a) follows from the law of total probability, (b) from \revOne{conditioning on the past state and on $Y^{n-1}$}, and (c) by exploiting the Markov property.
At every step, the \ac{APP} logarithmic ratio can be evaluated as
\begin{align}
	\lambda_n &= \ln \frac{\displaystyle \sum_{\sigma_n \in \{0\}\times \countset}P(\sigma_n,y^n)}{\displaystyle \sum_{\sigma_n \in \{1\}\times \countset} P(\sigma_n,y^n)}.
\end{align}}

\revTwo{We remark that the complexity of the calculation entailed by the recursion grows quadratically with the number of states in the \ac{HMM}, i.e., with $M^2$. We describe in the following a sub-optimal algorithm that computes approximate values of the \ac{APP} logarithmic ratios with a complexity that is independent of the number of sources. This makes it suited in cases in which $M$ is large.} 

\revTwo{\subsubsection{Myopic State Estimation}\label{sec:myopicest}
The recursive calculation implementing the optimal detector requires tracing probabilities over a trellis diagram with $2M$ states. A simplified approach consists in neglecting the effect of $S_{n-1}$ and $S_n$ on the probability of observing $Y_n$. This is equivalent to the derivation of a recursive state estimator where (a) the reference source adopts a reactive transmission approach and (b) all the remaining $M-1$ sources adopt a random approach with activation probability set to the stationary probability of a state change $\tilde{\alpha} := \pi_0 \tp{01} + \pi_1 \tp{10}$.
We refer to the estimator obtained under this approximation as \emph{myopic} estimator, and to the model described by conditions (a) and (b) as \emph{surrogate myopic model}.} 

\revTwo{For such model we have
$P(y_n | \sigma_{n}, \sigma_{n-1}) = P(y_n | x_n, x_{n-1})$,
hence 
\begin{align}
P(x_n,y^n)&\!=\!\!\!\! \sum_{x_{n-1}\in\sourcealpha}  \!\!\!\! P(y_n | x_n, x_{n-1}) P(x_n|x_{n-1})P(x_{n-1},y^{n-1})
\label{eq:jointProb_HMM_reactive_myopic}
\end{align}
resulting in the myopic \ac{APP} logarithmic ratio recursion \eqref{eq:APPLRMy} reported at the bottom at the page. \setcounter{equation}{14}

Here, (a) follows by a recursive application of \revOne{conditioning} and of the Markov property, whereas (b) exploits again the fact $P(x_{n-1}|y^{n-1})=P(x_{n-1}|\tilde\lambda_{n-1})$, which holds true in the surrogate myopic model, and $ P(x_{n-1}|\tilde\lambda_{n-1}) \propto \exp(-x_{n-1}\tilde\lambda_{n-1})$.
 We will see numerically that the recursion $\tilde{\lambda}_n = g(y_n,\tilde{\lambda}_{n-1})$ yields estimates of the actual \ac{APP} logarithmic ratio that are accurate enough to characterize the estimation entropy under the reactive transmission strategy with good approximation. Notably, the evaluation of the recursion $\tilde{\lambda}_n = g(y_n,\tilde{\lambda}_{n-1})$ entails a complexity that is independent on the number of sources.}

 \revTwo{\begin{remark}
     It is worth mentioning that, as a consequence of Remark \ref{rem:symmetric_TP},  the recursion \eqref{eq:APPLRMy} yields the exact \ac{APP} logarithmic ratio when the sources have symmetric transition probabilities, i.e., when $\tp{00}=\tp{11}$. 
 \end{remark}}

\section{Density Evolution Analysis} \label{sec:de}

Recalling that the distribution of the estimation entropy $H_n$ can be derived from that of the \ac{APP} logarithmic ratio using \eqref{eq:ent_yn}, we consider next the problem of obtaining the distribution of the \ac{r.v.} $\Lambda_n$. To do so, we employ a \ac{DE} \cite{RU01a} approach to the recursive calculation of the \ac{APP} logarithmic ratio density over the trellis diagram describing the evolution of the \ac{HMM} state \cite[Chapter 4]{Pfister2003}. \revOne{In particular, quantized \ac{DE} \cite{JinQDE} provides an efficient (i.e., fast) means to evaluate the estimation entropy $H_n$, and it represents a viable alternative to the use of long Monte Carlo simulations.} We instantiate the analysis for both the random and the reactive transmission strategies.  Note that in the reactive case, the analysis requires tracking the evolution of the joint distribution of $2M-1$ random variables, rendering the quantized \ac{DE} analysis intractable even under simple quantiziation rules. For this reason, we will resort only to the myopic state estimator.

\subsection{Random Transmission Strategy}
The analysis is based on a recursive calculation of the distribution of $\Lambda_n$ given the distributions of $\Lambda_{n-1}$ and of $Y_n|X_n$. Suppose the joint distribution of $\Lambda_{n-1}$ and $X_{n-1}$ to be known. We have that  
\begin{align}
&P(\lambda_n,x_n) = \sum_{x_{n-1}\in\mathcal{X}} \!\!\!\sum_{\substack{y_n,\lambda_{n-1}:\\ f(y_n,\lambda_{n-1})=\lambda_n}} \!\!\!\!\!\!\!\! P(\lambda_{n-1},y_n,x_n,x_{n-1})\\
&=\!\!\!\! \sum_{x_{n-1}\in\mathcal{X}} \!\!\!\sum_{\substack{y_n,\lambda_{n-1}:\\ f(y_n,\lambda_{n-1})=\lambda_n}} \!\!\!\!\!\!\!\!\!\!\!\! P(y_n|x_n,x_{n-1})P(x_n|x_{n-1})P(\lambda_{n-1},x_{n-1})\\
&=\!\!\!\!\!\!\!\! \sum_{\substack{y_n,\lambda_{n-1}:\\ f(y_n,\lambda_{n-1})=\lambda_n}} \!\!\!\!\!\! P(y_n|x_n) \!\!\sum_{x_{n-1}\in\mathcal{X}}  P(x_n|x_{n-1})P(\lambda_{n-1},x_{n-1}) \label{eq:DErand}
\end{align}
where $f(y_n,\lambda_{n-1})$ is given in  \eqref{eq:lambda_rnd},
which readily provides the evolution of the joint distribution. The recursion is initialized by assuming no initial knowledge on the state, i.e., by setting
$\prob{\Lambda_{-1}=0,X_{-1}=0}=\prob{\Lambda_{-1}=0,X_{-1}=1}=1/2$.

\subsection{Reactive Transmission Strategy}
As discussed, we work under the surrogate myopic model introduced in Sec.~\ref{sec:myopicest}. Also in this case the analysis is based on a recursive calculation of the distribution of $\tilde\Lambda_n$ given the distributions of $\tilde\Lambda_{n-1}$ and of $Y_n|X_n,X_{n-1}$. Suppose the joint distribution of $\tilde\Lambda_{n-1}$ and $X_{n-1}$ to be known. We have that  
\begin{align}
&P(\tilde\lambda_n,x_n) = \sum_{x_{n-1}\in\mathcal{X}} \!\!\!\!\sum_{\substack{y_n,\tilde\lambda_{n-1}:\\ g(y_n,\tilde\lambda_{n-1})=\tilde\lambda_n}} \!\!\!\!\!\!\! P(\tilde\lambda_{n-1},y_n,x_n,x_{n-1})\\
 &= \!\!\!\!\sum_{x_{n-1}\in\mathcal{X}} \!\!\!\!\!\!\sum_{\substack{y_n,\tilde\lambda_{n-1}:\\ g(y_n,\tilde\lambda_{n-1})= \tilde\lambda_n}} \!\!\!\!\!\!\!\!\!P(y_n|x_n,x_{n-1})P(x_n|x_{n-1})P(\tilde\lambda_{n-1},x_{n-1}) \label{eq:DEreact}
\end{align}
where $g(y_n,\tilde{\lambda}_{n-1})$ is given in  \eqref{eq:APPLRMy},
which  provides the evolution of the joint distribution. The recursion is initialized by setting
$\mathsf{P}[\tilde\Lambda_{-1}=0,X_{-1}=0]=\mathsf{P}[\tilde\Lambda_{-1}=0,X_{-1}=1]=1/2$.

\section{Results and Discussion} \label{sec:results}

We analyze the \ac{SEE} as a function of the nodes' population size under both random and reactive transmission strategies. The results are obtained via \ac{DE} analysis as reported in Sec.~\ref{sec:de}, and verified by means of Monte Carlo simulations. In the latter case, the entropy $\ent(y^n)$ is tracked relying on a MAP estimator for each realization, and the SEE is estimated by averaging the results obtained for large values of $n$.

\begin{figure}[t]
    \centering
    \includegraphics[width=.95\columnwidth]{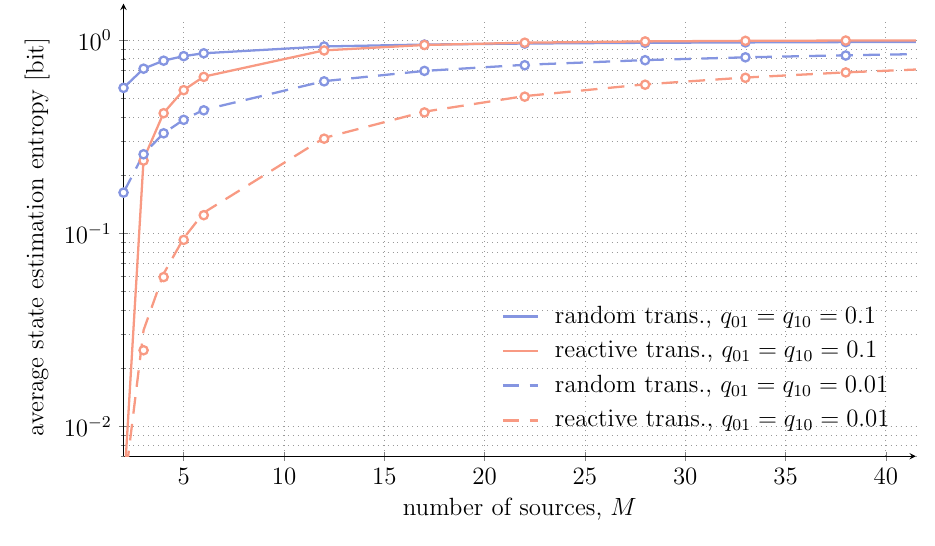}\vspace{-3mm}
    \caption{Average \ac{SEE} vs. number of nodes $M$, in the case of symmetric sources ($\qzo = \qoz$). Lines denote results obtained via \ac{DE} analysis, whereas markers the output of Monte Carlo simulations.}
    \label{fig:seeSymmetric}
\end{figure}

First insights on the behavior of the different access policies are provided in Fig.~\ref{fig:seeSymmetric}, which reports the average \ac{SEE} against the number of sources $M$ in the system for the symmetric case $\qzo=\qoz$. In the plot, blue lines refer to the random transmission approach, whereas red ones are relative to the reactive strategy. Solid and dashed patterns are used to distinguish results obtained for $\qzo=\qoz = 0.1$ and $\qzo = \qoz = 0.01$, respectively.

For the random transmission strategy, the activation probability $\alpha$ has been set to $1/M$. This choice maximizes the throughput of slotted ALOHA, and, as discussed in Sec.~\ref{sec:aoi_metric}, is also optimal in terms of average \ac{AoI}. It is easy to observe that setting $\alpha=1/M$ minimizes the \ac{SEE}, too.\footnote{\revOne{A similar trend \--- i.e., AoI as proxy to minimize the uncertainty on the source state \--- was also noted in \cite{sun_poly_TIT_2020,Yates20_Survey}, for a single source setting when sampling is performed independently of the source evolution.}} To see this, it is sufficient to note that the problem of estimating the state $X_n$ given the observations $Y^n$ is equivalent to the problem of estimating $X_n$ with $Y_0,Y_1,\ldots,Y_n$ being the observations of $X_0, X_1, \ldots, X_n$ after transmission over $n+1$ independent \acp{BEC} with erasure probability $\epsilon = 1-\omega$, where the channel output alphabet is $\{0,1,?\}$ and $?$ denotes an erasure. The observation follows by the fact that, under random transmissions, observing $Y_n\in\{\idle,\collision,\izero,\ione\}$ does not yield any information on $X_n$, hence any channel output in $\{\idle,\collision,\izero,\ione\}$ can be regarded as an erasure.  The conditional entropy $\ENT(X_n|Y^n)$ is hence minimized by minimizing $\epsilon$, i.e., by maximizing $\omega=\alpha(1-\alpha)^{M-1}$.

Fig.~\ref{fig:seeSymmetric} offers several take-aways. First, as expected, the average \ac{SEE} raises in all cases when more nodes populate the network. The trend stems from the harsher channel contention experienced for larger values of $M$, which increases the probability of losing updates due to collisions and thus the uncertainty at the receiver. 
In contrast, lower values of the state transition probability $\qzo=\qoz$ improve the \ac{SEE}. The reason is twofold. On the one hand, when the source status changes less often, fewer updates are required on average at the receiver to track its evolution, and the loss of packets entails a lower increase in the uncertainty level. On the other hand, more sporadic transitions reduce the channel contention in the reactive case, increasing the probability of successfully delivering a packet and positively impacting the metric. The plot also pinpoints the beneficial effect of implementing a reactive transmission strategy. As discussed, having sources only notifying state changes makes transmissions more informative and helps preventing congestion, whereas a random transmission approach may see nodes occupy the channel to send information already available at the receiver or, similarly, not promptly notify a relevant transition. The effect is especially apparent for low values of $M$. Notably, for $M=2$, perfect knowledge is available at the receiver for the reactive policy (\ac{SEE} equal to $0$). Indeed, in this case, once the state of both nodes is known, if a collision occurs the receiver deduces a simultaneous state change at the sources, which allows to infer the new states.

\begin{figure}[t]
    \centering
    \includegraphics[width=.95\columnwidth]{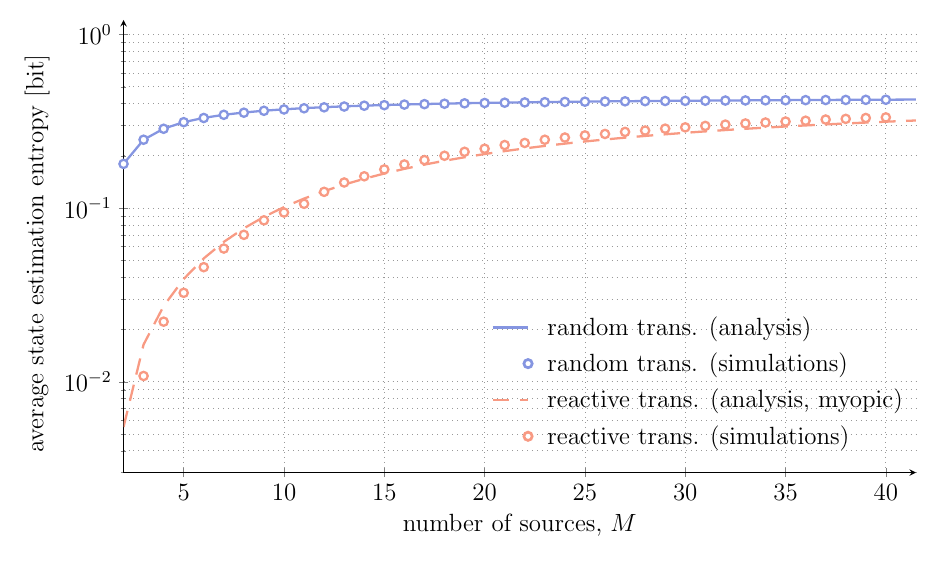}\vspace{-3mm}
    \caption{Average \ac{SEE} vs. number of nodes $M$. Asymmetric source case: $\qzo = 0.01$, $\qoz=0.1$. Lines denote results obtained via \ac{DE} analysis, whereas markers the output of Monte Carlo simulations. For the reactive case, the myopic surrogate model was used for \ac{DE}.}
    \label{fig:seeAsymmetric}
\end{figure}

The behavior in presence of asymmetric sources is reported in Fig.~\ref{fig:seeAsymmetric}, considering \mbox{$\qzo=0.01$} and $\qoz=0.1$. In this case, we recall that the analytical results obtained via \ac{DE} for the reactive transmission strategy offer an approximation, as they were derived resorting to the myopic surrogate model described in Sec.~\ref{sec:myopicest}. Nonetheless, a very tight match can be observed with the results of Monte Carlo simulations, which estimate the average \ac{SEE} taking into account the  evolution of all sources in the systems. The outcome is particularly interesting, as it corroborates the accuracy of the proposed simplified analytical approach in capturing the behavior of the policy also for asymmetric transitions.

For the rest, the plot confirms the trends discussed in the symmetric case. From this standpoint, it is interesting to observe that the average \ac{SEE} tends to converge for large values of $M$ to the entropy of the stationary state distributions of the respective Markov chains, which is $1$ for the symmetric case and $\approx0.44$ for the asymmetric one under study. In fact, as more nodes populate the system, the number of received packets per source progressively sinks (either because of collisions, in the case of the reactive scheme, or because the activation probability falls to $0$ asymptotically in $M$ for the random scheme), reducing the amount of information available at the receiver on the tracked processes.

\begin{figure}[t]
    \centering
    \includegraphics[width=.95\columnwidth]{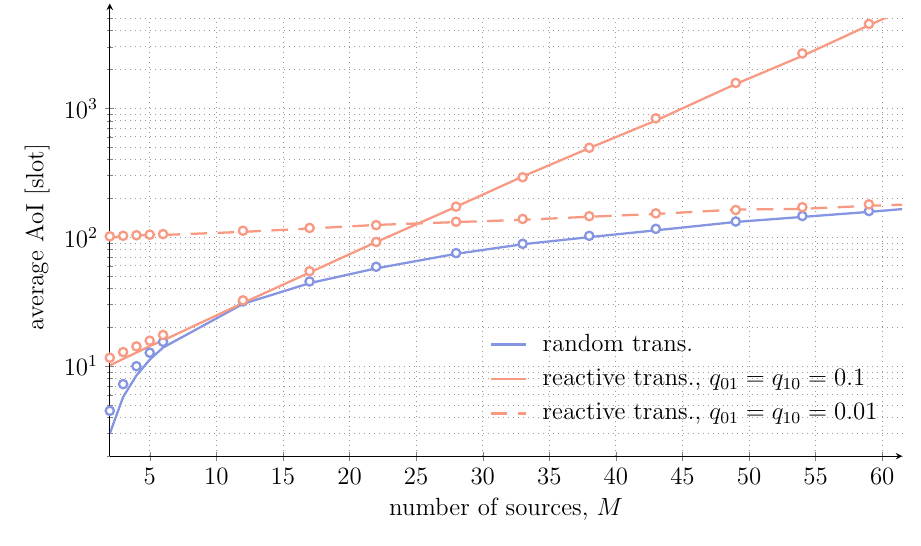}\vspace{-3mm}
    \caption{Average AoI vs number of nodes $M$, in the case of symmetric sources. Lines denote results obtained analytically, whereas markers the output of Monte Carlo simulations. For the random strategy, $\alpha=1/M$.}
    \label{fig:aoiSymmetric}
\end{figure}

The results reported so far have highlighted the beneficial role played by a transmissions strategy that is tuned to the process being monitored, when aiming at maintaining a low SEE at the receiver in the presence of sources that are not memoryless. Such outcome is particularly interesting, as it suggests medium access control design principles that inherently differ from those commonly considered when targeting information freshness. To appreciate this, we explore  the behavior of both the random and reactive transmission strategies in terms of average AoI. As discussed in Sec.~\ref{sec:aoi_metric}, the metric is commonly employed to capture how up to date the perception of a monitored process is at the receiver, and tracks the time elapsed since the generation of the last received update. In slotted ALOHA systems, AoI takes the form reported in \eqref{eq:aoi_SA}, and is minimized by maximizing the average throughput, i.e., the frequency with which each source can successfully report data. From this standpoint, we recall that the metric is by definition oblivious of the value being delivered, so that generating and delivering a reading leads to a reset of AoI regardless of the actual content of the message.

The average AoI obtained for the access strategies under study is reported in Fig~\ref{fig:aoiSymmetric} against the number of nodes in the network, considering symmetric sources ($\qoz=\qzo=0.1$, and $\qoz=\qzo=0.01$). For the random transmission, $\alpha=1/M$, as already discussed. We further note that only one curve is reported for such approach, as its performance in terms of AoI does not depend on the transitions of the underlying monitored source. This is not the case for reactive transmissions, as the activation probability inherently depends on how frequently the sources transition. \revTwo{Focusing on the two curves for the reactive approach, we also note that, when few nodes are present, a higher transmission frequency ($\qzo=0.1$) leads to lower AoI, as collisions are seldom experienced and updates can be delivered more often. As soon as $M>10$, however, the channel becomes overloaded (average load larger than $1$ pkt/slot, and the AoI rapidly deteriorates. In this region, the lower transmission rate experienced due to $\qzo = 0.01$ becomes beneficial, avoiding excessive congestion.}

More interestingly, the plot reveals that the average AoI attained with a random transmission strategy is always lower compared to the one offered by the reactive scheme. The only point in which the two strategies coincide corresponds to the situation in which the state transition probability coincides with the optimal activation probability (i.e., $\tilde\alpha = 1/\qzo = 1/M$, obtained for $M=10$ when $\qzo=0.1$ and for $M=100$ when $\qzo = 0.01$ \--- out of the plot).

\am{Remarkably, AoI and SEE suggest the use of different access solutions.  The intuition behind this behavior is that AoI treats all packets as equally informative, and taking the stochastic model of the source into account does not provide any advantage. However, when the actual state of a monitored process plays a major role, maintaining a low uncertainty level at the receiver may be critical. In this sense, the SEE naturally emerges as a good candidate metric, and the profoundly different hints it provides in terms of access strategy shall be taken into account in the design of the system. In turn, AoI may be the metric of choice when fresh information is needed, or can be a valuable proxy when the source statistics are not known at the receiver.}

\begin{figure}[t]
    \centering
    \includegraphics[width=.95\columnwidth]{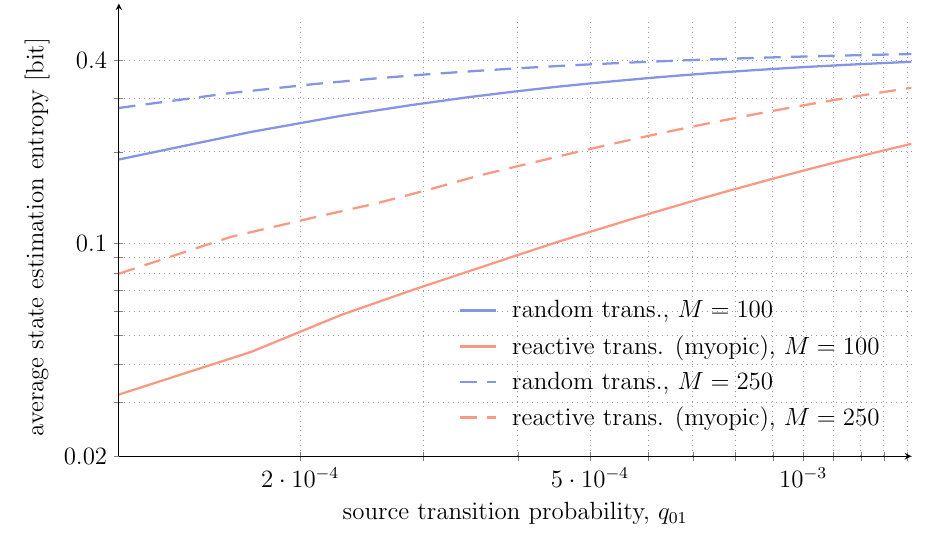}\vspace{-3mm}
    \caption{\revTwo{Average SEE vs. source transition probability $\qzo$, obtained by means of density evolution (myopic approximation for the reactive case). Solid lines report results for $M=100$, whereas dashed ones for $M=250$. In all cases, $\qoz = 10 \qzo$. Accordingly, for the reactive strategy the $x$-axis denotes a span of channel load from $0.02$ to $0.28$ [pkt/slot] when $M=100$ and from $0.05$ to $0.7$ [pkt/slot] when $M=250$.}}
    \label{fig:see_more_nodes}
\end{figure}

\revTwo{To conclude our study, we report in Fig.~\ref{fig:see_more_nodes} the \ac{SEE} attained with the different transmission strategies for larger user populations, representative of practical IoT settings. To this aim, we set $M=100$ (solid lines) or $M=250$ (dashed lines), and consider asymmetric sources with $\beta:= \qoz/\qzo = 10$. Accordingly, the stationary probabilities evaluate to \mbox{$\pi_0=\beta/(1+\beta) \approx 0.909$} and \mbox{$\pi_1 = 1/(1+\beta) \approx 0.091$}. In the plot, the $x$-axis explores different values of $\qzo$, denoting an increasing transition rate for the sources. For the random strategy, we set in all cases $\ptx=1/M$, so that the system operates at a channel load of $1$ packet per slot. For the reactive strategy, instead, we observe that the average load can be estimated (under the myopic approximation) as $M(\pi_0 q_{01} + \pi_1 q_{10})  = 2M\beta\qzo/(1+\beta)$, and the $x$-axis of Fig.~\ref{fig:see_more_nodes} can also be interpreted as a scaled version of the contention level (see figure caption). In the considered setting, the stationary entropy of the source is $\ENT(X) = 0.4395$, as was in Fig.~\ref{fig:seeAsymmetric}. As expected, an increase in the number of nodes leads to higher values of SEE for a given $\qzo$. For the reactive strategy, this is a direct consequence of the increased level of contention, and thus the lower success probability. As far as the random strategy is concerned, although the load is kept constant, and so is the success probability ($e^{-1} \approx 0.37$), the higher SEE stems from the more sporadic access opportunities each node has ($\ptx=1/M$). We observe how, also for larger values of $M$, the beneficial effect of tying the access policy to the source evolution is apparent when it comes to reducing the uncertainty at the receiver, as can be appreciated from the remarkable gap between the random and reactive strategies.
}
\section{Conclusions} \label{sec:conclusions}

In this paper, we have studied a system in which multiple terminals share a common slotted ALOHA channel to report updates to a receiver.  Assuming each node to monitor a two-state Markov source, we characterized the performance of the system in terms of average state estimation entropy, capturing the uncertainty at the receiver about the state of the tracked processes, under two transmission strategies: random and reactive. In the former case, a node randomly sends a status update at each slot, whereas in the latter a message is transmitted only if the monitored source has experienced a state change. We provided an analytical characterization of the SEE, and showed that its calculation is amenable to efficient implementation through DE, which allows to evaluate how the system performance scales with the number of source nodes.
Our study reveals that a reactive solution can offer better performance in terms of SEE, lowering channel congestion and favoring delivery of relevant updates. Notably, this design hint differs from what suggested when considering the average AoI as reference metric, for which the random transmission approach is convenient.  From this standpoint, AoI is an adequate metric in contexts where limited knowledge about the source statistics is available, and  it is reasonable to assume that drought of updates translate into higher uncertainty about the state of the source. However, in setups where the receiver has knowledge of the source model, SEE naturally emerges as a good metric, as it captures the residual uncertainty at the receiver about the state of the sources once the channel output and the source model have been taken into account.

\begin{appendices}
\section{Pragmatic State Estimation: Decode\&Hold} \label{sec:dh}

The approach presented in Sec.~\ref{sec:optimumEstimation} provides the receiver with an optimal estimate of the current source state, minimizing the probability of error. The complexity entailed by running a \ac{MAP} estimator may however be critical in settings where messages are delivered to a battery-powered and computationally-limited collector \cite{LoRa,Wang17_CommMag}. In addition, it requires knowledge on the statistics of the source, which may not be available. For such scenarios, other detectors may be preferred, trading off an optimal estimate in favor of a simpler implementation. 

Starting from this remark, we consider an alternative solution, based on a \emph{decode and hold} (D\&H) estimator. In this approach, the receiver maintains at any time $n$ an estimate $\est_n$ for the reference source, which is updated whenever a packet is decoded and reveals the current state of the node of interest, or remains unchanged otherwise, i.e.
\begin{align}
	\est_n = 
	\begin{cases}
		Y_n 		& \text{if } Y_n \in\{0,1\}\\
		\est_{n-1}	& \text{if } Y_n \in\{\idle,\collision,\izero,\ione\}\,.
	\end{cases}
\label{eq:DH_def}
\end{align}
From \eqref{eq:DH_def} we see that the D\&H solution does not require the calculations of an APP logarithmic ratio as in the MAP estimator, nor does it lean on knowledge of the source statistics. On the other hand, the following example provides an intuition of why this simple scheme is inherently suboptimal.
\begin{example}
	Consider the case of a system with $M=2$ sources, operating under the reactive transmission strategy, and assume the following evolution:
	\begin{itemize}
		\item at time $n-2$, we have $X_{n-2} = 0$, $X^{(1)}_{n-2} = 1$
		\item at time $n-1$ only the source of interest transitions: $X_{n-1} = 1$, $X^{(1)}_{n-1} = 1$
		\item at time $n$ both sources transition: $X_{n} = 0$, $X^{(1)}_{n} = 0$
	\end{itemize} 
	Accordingly, slot $n-1$ sees the sole transmission of the reference source, so that $Y_{n-1}=1$, whereas a collision is experienced over slot $n$, i.e. $Y_n = \collision$. In this situation, the D\&H estimator outputs the sequence $\est_{n-1} = 1$, $\est_{n}=1$, providing an erroneous estimate in slot $n$ (i.e., $\est_n \neq X_n$) Conversely, the output of the threshold test on the APP logarithmic ratio in \eqref{eq:lambdan_def} performed by the MAP approach returns the correct estimates in both time instants. This can readily be verified by observing that, for the case under study, \mbox{$\mathsf P[X_{n}=0 \,|\, Y_n=\collision,Y_{n-1}=1,Y^{n-2}=y^{n-2}\,] = 1$}. Indeed, the observation of a collision implies that both nodes transmitted, and hence changed their state, providing certain knowledge of the state of the reference source at time $n$ as well. 
\end{example}

To characterize the performance of this low-complexity solution, we focus on the average error probability $P_e$, which can be effectively computed by jointly tracking the source state and estimate processes via the Markov chain $(X_n,\est_n)$. The analysis we present is exact under random transmissions, whereas it resorts to the myopic surrogate model when the reactive strategy is implemented.

\begin{figure*}
	\centering
	\subfloat[random transmission strategy]{
	\includegraphics[width=.42\textwidth]{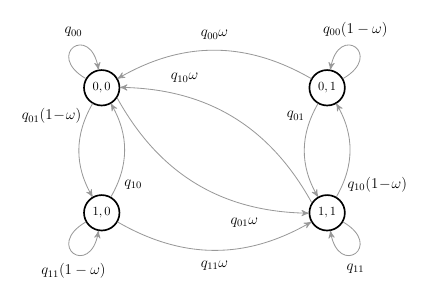}
	\label{fig:mc_est_random}
	}
	\hspace{-1.5em}
	\subfloat[reactive transmission strategy]{
		\includegraphics[height=.29\textwidth]{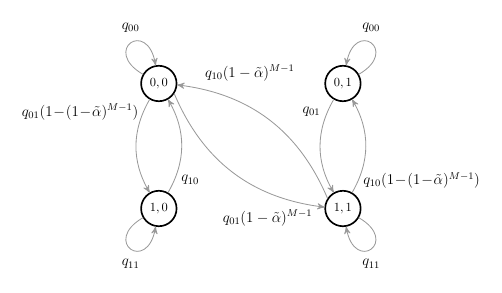}
		\label{fig:mc_est_reactive}
	}
	\caption{Markov chains $(X_n,\est_n)$ tracking the evolution of reference source state and D\&H estimate, in the case of a random (a) and reactive (b) transmission policy.}
\end{figure*}

Let us first focus on the random approach. In this case, the transition probabilities of the chain are summarized in Fig.~\ref{fig:mc_est_random}, where we recall that $\pup = \alpha (1-\alpha)^{M-1}$ is the probability that a node  successfully delivers an update over a slot. As an example, consider state \mbox{$(X_n,\est_n)=(0,0)$}. The process remains in the same state over the next slot if no change of state occurs, i.e. with probability $\qzz$. Note indeed that $\est_{n+1}$ will remain $0$ both in case of transmission and successful delivery of an update (i.e. $Y_{n+1}=0$), and in the absence of a received packet from the reference source (i.e. $Y_{n+1}\in\{\idle,\collision,\izero,\ione\}$). Instead, no transition to state $(0,1)$ can take place, since the D\&H estimator would only reset $\est_{n+1}$ to $1$ upon receiving an update from the source containing that value, which is not possible when $X_{n+1}=0$. In turn, the system moves to $(1,0)$ \--- providing an erroneous estimate of the source state \--- whenever the node of interest changes state (probability $\qzo$) but does not deliver an update, either due to a collision or for lack of transmission (overall probability $1-\pup$). Conversely, a transition to $(1,1)$ occurs when the source moves to state $1$ and successfully sends a packet in slot $n+1$ (probability $\qzo \pup$). All other probabilities in the chain can be derived following a similar reasoning.

The finite state Markov process is readily shown to be aperiodic and irreducible, and thus ergodic. Accordingly, the error probability $P_e$ introduced in \eqref{eq:Pe_def}, expressing the average time spent by the chain in $(0,1)$ and $(1,0)$, can be computed as the sum of the stationary probabilities of such states, denoted by $\pi_{(0,1)}$ and $\pi_{(1,0)}$. Solving the balance equation, we get 
\begin{align}
	P_e = \pi_{(0,1)} + \pi_{(1,0)} = \frac{2 \qzo \qoz \left(1-\pup \right) }{(\qzo + \qoz)\left[\,\pup + (1-\pup)(\qzo + \qoz)\right]}\,.
	\label{eq:Pe_DH_random}
\end{align}
The same approach can be leveraged to derive the performance of the D\&H estimator when the nodes operate following a reactive transmission policy, leaning on the surrogate myopic model introduced in Sec.~\ref{sec:myopicest}.  The corresponding transition probabilities for the Markov chain $(X_n, \est_n)$ take the form reported in Fig.~\ref{fig:mc_est_reactive}. In this case, the term $(1-\tilde\alpha)^{M-1}$ captures the probability for a source to deliver an update over a slot once a state change has taken place, with the activation probability $\tilde\alpha$ defined in \eqref{eq:myopic_def}. The stationary distribution of the chain gives in this case
\begin{align}
	\pi_{(0,1)} &\approx \frac{\qoz (1-(1-\tilde\alpha)^{M-1})}{(\qzo+\qoz)(2-(1-\tilde\alpha)^{M-1})}\\[.2em]
	\pi_{(1,0)} &\approx \frac{\qzo (1-(1-\tilde\alpha)^{M-1})}{(\qzo+\qoz)(2-(1-\tilde\alpha)^{M-1})}
\end{align}
leading to an average error probability
\begin{align}
	P_e \approx \dfrac{1-(1-\tilde\alpha)^{M-1}}{2-(1-\tilde\alpha)^{M-1}} \,.
	\label{eq:Pe_DH_reactive}
\end{align}
As discussed, this is exact for symmetric sources (i.e. $\qzo=\qoz$), whereas it is only an approximation in the asymmetric case. 

\begin{figure}
	\centering
	\includegraphics[width=.95\columnwidth]
    {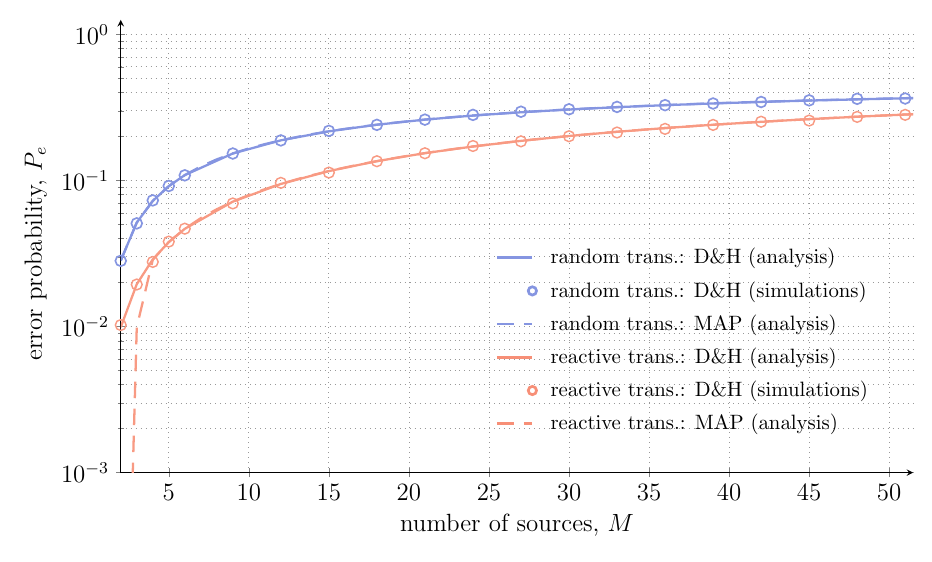}
    \vspace{-5mm}
	\caption{Average state estimate error probability, $P_e$, against the number of devices in the network, symmetric source case ($\qoz=\qzo=0.01$). Solid lines denote analytical results for the D\&H estimator, whereas circle markers the outcome of Monte Carlo simulations. Dashed lines report  results for a MAP estimator using \eqref{eq:PEMAPRAND} and \eqref{eq:PEMAPREACT}. Different colors indicate the performance attained under the random (blue) or reactive (red) transmission policy.} 
	\label{fig:Pe_symmetric}
\end{figure}

First insights on the behavior of the D\&H estimator are offered by Fig.~\ref{fig:Pe_symmetric}, which reports $P_e$ against the number of nodes in the network in the case of symmetric sources (\mbox{$\qzo=\qoz=0.01$}). Blue lines refer to performance attained under the random transmission policy, whereas red ones are representative of the reactive approach. In the former case, the activation probability has been set as the reciprocal of the number of nodes, i.e., $\alpha=1/M$, in order to maximize the throughput and thus the average number of delivered updates. In the plot, solid lines report the analytical results for the D\&H estimator obtained via \eqref{eq:Pe_DH_random} and \eqref{eq:Pe_DH_reactive}, whereas circle markers denote the results of Monte Carlo simulations. Finally, dashed lines  show the performance of a MAP estimator, which can be derived from the \ac{DE} analysis. In particular, for the random transmission strategy case the error probability is
\begin{equation}
 P_e = \lim_{n \to \infty} \sum_{\lambda_n \leq 0} P(\lambda_n,0) + \lim_{n \to \infty} \sum_{\lambda_n > 0} P(\lambda_n,1)   \label{eq:PEMAPRAND}
\end{equation}
where $P(\lambda_n,x_n)$  is computed with the recursion \eqref{eq:DErand}. Note that to practically estimate the limit in \eqref{eq:PEMAPRAND} it suffices to let $n$ grow large enough (e.g., $n\approx 10^5$) to observe converging probability estimates. Similarly, the error probability for the \ac{MAP} estimator for the reactive transmission strategy case can be obtained, in the myopic approximation setting, as
\begin{equation}
 P_e \approx \lim_{n \to \infty} \sum_{\tilde\lambda_n \leq 0} P(\tilde\lambda_n,0) + \lim_{n \to \infty} \sum_{\tilde\lambda_n > 0} P(\tilde\lambda_n,1)   \label{eq:PEMAPREACT}
\end{equation}
where $P(\tilde\lambda_n,x_n)$ follows from the \ac{DE} recursion \eqref{eq:DEreact}. 

The reported trends pinpoint a visible gap between the two estimators when few nodes populate the network. The rationale behind this goes along the lines of the discussion presented in Example 2. Indeed, while both MAP and D\&H attain an exact knowledge whenever a packet from the tracked source is received, the former refines its estimate also in the presence of an idle slot, a collision, or upon receiving a packet from another node ($Y_n\in\{\idle,\collision,\izero,\ione\}$). Such side information is especially beneficial for low values of $M$, as it allows to infer with a good level of confidence the state of the reference process. A simple quantitative intuition on this can be grasped by focusing on the reactive strategy and by considering the likelihood ratio 
\begin{align}
	&\dfrac{\mathsf P[ X_{n} = X_{n-1}, Y_n = \collision \,]}{\mathsf P[ X_{n} \neq X_{n-1}, Y_n = \collision \,]} \\
	&= \frac{(1-\alpha) \left[1-(1-\alpha)^{M-1} - (M-1)(1-\alpha)^{M-2}\right]}{\alpha\left[1-(1-\alpha)^{M-1}\right]} \hspace*{1em}
	\label{eq:pColl_reactive}
\end{align}
obtained in the event of a collision and only looking at the outcome of the last slot. The quantity evaluates to $0$ for $M=2$, allowing the MAP estimator to extract exact knowledge on the change of state, as discussed in the previously presented example. 

Interestingly, the performance gap between the two approaches vanishes as the number of sources increases. For larger $M$, indeed, idle slots occur more seldom, and the impact of observing a collision on the MAP estimate of the reference source becomes weaker. A hint on this is again offered by \eqref{eq:pColl_reactive}, as the likelihood ratio converges to $(1-\alpha)/\alpha$ for $M\rightarrow\infty$. This observation is of practical relevance, suggesting that the simple D\&H solution offers good performance in sufficiently large networks when symmetric sources are to be tracked. 

Fig.~\ref{fig:Pe_symmetric} also reveals that a lower error probability is attained for the configuration under study when nodes implement a reactive transmission approach, especially for low to intermediate values of $M$. The choice of accessing the channel only to signal a change of state is in this case particularly beneficial, increasing the probability of successfully notifying an update. Conversely, when sources are sampled at random times, nodes may attempt to report information which is already available at the receiver, congesting the medium unnecessarily and generating additional collisions that reduce the estimator accuracy.

\begin{figure}
	\centering
	\includegraphics[width=.95\columnwidth]{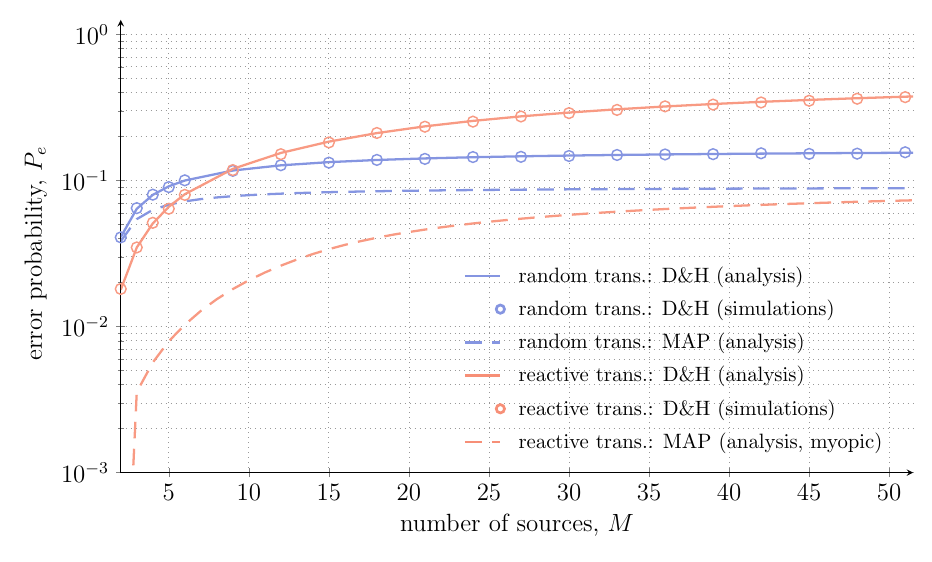}
 \vspace{-5mm}
	\caption{Average state estimate error probability, $P_e$, against the number of devices in the network, asymmetric source case ($\qoz=0.01, \qzo=0.1$). Solid lines denote analytical results for the D\&H estimator, whereas circle markers the outcome of Monte Carlo simulations. Dashed lines report results for a MAP estimator obtained via \eqref{eq:PEMAPRAND} and \eqref{eq:PEMAPREACT} (myopic surrogate). Different colors indicate the performance attained under the random (blue) or reactive (red) transmission policy.}
	\label{fig:Pe_asymmetric}
\end{figure}

These remarks are complemented by Fig.~\ref{fig:Pe_asymmetric}, which shows the same set of performance trends in the case of asymmetric sources, assuming $\qzo = 0.01, \qoz = 0.1$. Within the plot, let us first consider the behavior observed when the reactive policy is implemented (red lines). In this case, simulation results show a very tight match with the analytical formulation of $P_e$ for the D\&H estimator \eqref{eq:Pe_DH_reactive}, obtained relying on the myopic surrogate approximation. On the other hand, a significant gap is present with respect to the behavior of a MAP approach, for all values of $M$. 
This stems from the long time that may be required for an erroneous D\&H estimate to be corrected. For the configuration under study, for instance, a change of the Markov chain $(X_n,\est_n)$ in  Fig.~\ref{fig:mc_est_reactive} from state $(0,1)$ to an exact knowledge can occur at the earliest when the source takes the less likely transition $0\rightarrow 1$ and the transmitted update is correctly received (i.e., $(X_n,\est_n)$ transitions to $(1,1)$, taking on average $1/[\qzo(1-\tilde\alpha)^{M-1}]$ slots). Conversely, the MAP approach can better refine its estimate at each slot, possibly correcting the erroneous knowledge without the need for the source to perform any further transmission. 
In addition, and in contrast to what observed in the symmetric case, the D\&H estimator performs worse when nodes implement a reactive rather than a random transmission policy, already for relatively low values of $M$. The trend can again be explained observing that in the former case the estimator can remain in error for a long time due to lack of transitions (and thus update transmissions) of the reference source. When delivery attempts are performed at random times, instead, such periods of drought can be shortened, with beneficial effects on the average error probability.

\section{Conditional Statistics of Channel Observations} \label{sec:appendix}

In this appendix, we derive the statistical relation between the observed channel output at a generic slot, $Y_n$, and the state of the system sources, as introduced in the \acp{HMM} of Sec.~\ref{sec:optimumEstimation}.

\subsection{Random Transmission Strategy}

In this case, we are interested in computing the conditional distribution $P(y_n|x_n)$ in \eqref{eq:condProb_HMM_random}. We start by observing that, by the i.i.d. behavior of all nodes across slots, the channel output does not depend on the state of the reference source when $Y_n\in\{\idle,\collision,\izero,\ione\}$, i.e. in the case of idle slot, collision, or message reception from another node. Accordingly, we get
\begin{align}
\label{eq:condProb_rand}
    \,\,\mathsf P[Y_n = \idle \given X_n=x_n] &=  (1\!-\!\alpha)^{M} \\
    \,\,\mathsf P[Y_n = \collision \given X_n=x_n] &= 1 - (1\!-\!\alpha)^{M} - M\pup \\
    \,\,\mathsf P[Y_n = \izero \given X_n=x_n] &= (1-\alpha) (M-1) \pi_0\, \pup\\
    \,\,\mathsf P[Y_n = \ione \given X_n=x_n] &= (1-\alpha) (M-1) \pi_1\, \pup\,.
\end{align}

In particular, an idle slot is experienced when none of the $M$ sources becomes active, whereas a collision occurs when more than one node transmits. In turn, the probability of receiving a "one" or "zero" message from a source other than the reference one is obtained by jointly considering the event of having the reference source not transmitting ($1-\alpha$), a single packet sent over the slot by one of the other nodes (probability $(M-1)\pup$), and that the sender is in the corresponding state (probability $\pi_0$ or $\pi_1$). 

Finally, the conditional probabilities when a packet from the reference source is received can be obtained as \mbox{$\mathsf P[Y_n=0\given X_n=0] = \mathsf P[Y_n=1\given X_n=1] = \pup$},  \mbox{$\mathsf P[Y_n=0\given X_n=1] = 0$} and \mbox{$\mathsf P[Y_n=1\given X_n=0] = 0$}. In the first case, the outcome is observed when the source transmits while all other contenders remain silent, whereas receiving a message from the source with a state different from the current one is not possible.

\subsection{Reactive Transmission Strategy, Symmetric Sources}

When sources are symmetric ($\qzo=\qoz)$ and a reactive transmission strategy is employed, the corresponding HMM is fully characterized by specifying the conditional probabilities $P(y_n\given x_{n-1},x_n)$, as highlighted in \eqref{eq:jointProb_HMM_reactive_myopic}. Consider first the case $Y_n=\idle$. Such an outcome can only be observed if the reference source does not transition (i.e., $X_{n-1}= X_n)$ and the same holds for all other nodes. We thus get $\mathsf P[ Y_n=\idle \given X_{n-1} =  x_{n-1},X_n=x_n] = (1-\tilde\alpha)^{M-1}$ for $x_{n-1} = x_n$ and $0$ otherwise. The reference source has to remain in the same state also for the receiver to retrieve a packet from any of the other nodes ($Y_n \in \{\izero,\ione\})$. Accordingly, 
\begin{align}
&\mathsf P[Y_n = \izero \given X_{n-1}=  x_{n-1},X_n=x_{n-1}] \\
&\hspace*{2em}= \mathsf P[Y_n = \ione \given X_{n-1} =  x_{n-1},X_n=x_{n-1} ]\\ 
&\hspace*{2em}= (M-1) \,\frac{\tilde\alpha}{2} \,(1\!-\!\tilde\alpha)^{M-2}
    \label{eq:pOtherGivenSame}
\end{align}
whereas $\mathsf P[Y_n=\izero|X_{n-1}\neq X_n] =\mathsf P[Y_n=\ione|X_{n-1}\neq X_n] = 0$. Within \eqref{eq:pOtherGivenSame}, the term $\tilde\alpha/2 = \pi_0 \qzo = \pi_1 \qoz$ denotes the probability for one of the $M-1$ sources to perform the transition which is successfully reported to the receiver. 

Similarly, when considering a collision outcome, the cases in which a state change (i.e., transmission) for the reference source takes place or not have to be distinguished. In the former, the activation of one or more of the remaining $M-1$ nodes suffices to have $Y_n=\collision$, whereas two ore more have to change state if the reference node does not. Following this reasoning we obtain
\begin{align}
    &\mathsf P[ Y_n = \collision \given X_{n-1}\!=\!x_{n-1}, X_n\! =\! x_n] \\
    &=\begin{cases}
        1 \!-\! (1\!-\! \tilde\alpha)^{M-1}   &\, \text{if } x_{n-1} \neq x_n \\
        1 \!-\! (1\!-\! \tilde\alpha)^{M-1}
        \!-\! (M\!-\!1)\,\tilde\alpha \,(1\!-\! \tilde\alpha)^{M-2} &\, \text{otherwise.} \\
    \end{cases}
\end{align}

Lastly, the conditional probability of observing a reading from the source of interest can be derived with the same reasoning applied in the random transmission case:
\begin{align}
    &\mathsf P[ Y_n = 0 \given X_{n-1}\!=\!x_{n-1}, X_n\! =\! x_n] \\[.2em]
    &\hspace*{4em}=\begin{cases}
        (1\!-\! \tilde\alpha)^{M-1}   &\, \text{if } x_{n-1}=1, x_n=0  \\
        0 &\, \text{otherwise}
    \end{cases}
    \label{eq:condProb_reactive_symm_0}
\end{align}

\begin{align}
    &\mathsf P[ Y_n = 1 \given X_{n-1}\!=\!x_{n-1}, X_n\! =\! x_n] \\[.2em]
    &\hspace*{4em}=\begin{cases}
        (1\!-\! \tilde\alpha)^{M-1}   &\, \text{if } x_{n-1}=0, x_n=1 \\
        0 &\, \text{otherwise.}
    \end{cases}
    \label{eq:condProb_reactive_symm_1}
\end{align}

\subsection{Reactive Transmission Strategy, Asymmetric Sources}
For the general case of asymmetric sources ($\qoz\neq\qzo)$, we are interested in computing both the one-step transition probabilities of the Markov chain $\sigma_n=(X_n,S_n)$ and the conditional probabilities $P(y_n\given\sigma_{n-1},\sigma_n)$. Let us first consider the former. Recalling the independent behavior of the reference source, we readily get
\begin{align}
    P(\sigma_n\given \sigma_{n-1}) = P(x_n \given x_{n-1}) \, P(s_{n} \given s_{n-1}).
\end{align}
Denote now for the sake of compactness as $\bar S_{n}$ the r.v. describing the number of sources, apart from the reference one, in state $1$ at a generic slot $n$, i.e.
\begin{align}
 \bar S_{n} = M - 1 - S_n \,.
\end{align}
By simple combinatorial arguments, it follows that
\begin{align}
    &\mathsf P[\,S_n = \sno + k \given S_{n-1} = \sno \,] \\
    &\hspace*{.3em}= \!\!\!\!\!\sum_{\ell=0}^{\min\{\sno,\sco\-k\}} \!\!\!\binom{\sno}{\ell} \qzo^\ell  \, \qzz^{\sno-\ell} \, \binom{\sco}{\ell+k} \qoz^{\ell+k} \, \qoo^{\sco - \ell - k}
\end{align}
for any $0 \leq k \leq M-1-\sno$. The expression accounts for all the possible cases in which the number of sources transitioning from state $1$ to $0$ is $k$ more than those changing from $0$ to $1$. Similarly, when the number of sources in state $0$ experiences an overall decrease, we obtain for any $1 < k \leq \sno$
\begin{align}
    &\mathsf P[\,S_n \!=\! \sno \!-\! k  \given S_{n-1} = \sno\,] \\
    &= \!\!\!\!\!\!\sum_{\ell=1}^{\min\{\sno-k,\sco\}} \!\!\!\binom{\sno}{\ell} \qzo^\ell  \, \qzz^{\sno-\ell} \, \binom{\sco}{\ell-k} \qoz^{\ell-k} \, \qoo^{\sco - \ell + k}.
\end{align}

Leaning on this, the conditional probabilities of observing $Y_n$ can be derived. Consider first the case $Y_n=\idle$. Recalling that an idle slot under the reactive strategy occurs only when neither the reference source nor any of the other nodes transition, we get 
for $S_n=S_{n-1}$ and $X_n=X_{n-1}$
\begin{align}
    &\mathsf P \left[ \, Y_n = \idle \given \sigma_{n-1} = (x_{n-1}, \sno ),\sigma_n=\sigma_{n-1}\,\right] \\[-.3em]
    &\hspace*{2em}= \dfrac{\qzz^{\sno} \, \qoo^{\sco}}{\mathsf P\left[ S_n = \sno \given  S_{n-1} = \sno \right]}
    \label{eq:condProb_reactiveAsymm_idle}
\end{align}
and $\mathsf P \left[ \, Y_n = \idle \given \sigma_{n-1} = (x_{n-1}, \sno ),\sigma_n=(x_n,s_n)\,\right] = 0$ otherwise. In \eqref{eq:condProb_reactiveAsymm_idle}, only the cases in which none of the other sources change state (probability $\qzz^{\sno} \, \qoo^{\sco}$) are accounted for in triggering an idle slot, as the overall event $S_{n}=S_{n-1}$ also includes all cases in which the same number of nodes transitions from $0$ to $1$ and from $1$ to $0$. Following a similar reasoning, the conditional probabilities for the receiver to decode a packet from a source different from the reference one follow. Specifically, for $S_n=S_{n-1}-1$ and $X_n=X_{n-1}$
\begin{align}
    &\mathsf P \left[ \, Y_n = \ione \given \sigma_{n-1} = (x_{n-1}, \sno ),\sigma_n=(x_{n-1},s_{n-1}-1)\,\right] \\[-1em]
    &\hspace*{2em}=\dfrac{\sno \,\qzo \,\qzz^{\sno-1}\, \qoo^{\sco-1}}{\mathsf P\left[ S_n = \sno-1 \given  S_{n-1} = \sno \right]}
\intertext{and, for $S_n=S_{n-1}+1$, $X_n=X_{n-1}$}
    &\mathsf P \left[ \, Y_n = \izero \given \sigma_{n-1} = (x_{n-1}, \sno ),\sigma_n=(x_{n-1},s_{n-1}+1)\,\right] \\[-.5em]
    &\hspace*{2em}=\dfrac{\sco \,\qoz \,\qoo^{\scn-1}\, \qzz^{\sco}}{\mathsf P\left[ S_n = \sno+1 \given  S_{n-1} = \sno \right]}.
\end{align}
In all other cases, the events cannot be observed. The expressions capture the event that only one of the sources performs a transition and notifies its new state.

Finally, the conditional probabilities for $Y_n$ to take value $0$ or $1$ are akin to those obtained in \eqref{eq:condProb_reactive_symm_0}, \eqref{eq:condProb_reactive_symm_1}, as only the reference node has to transition over slot $n$. In this case, accounting for the asymmetry of the other sources, we have for $S_n=S_{n-1}$ and $X_n=1$, $X_{n-1}=0$
\begin{align}
    &\mathsf P \left[ \, Y_n = 1 \given \sigma_{n-1} = (0, \sno ),\sigma_n=(1,\sno)\,\right] \\
    &\hspace*{2em}=\dfrac{ \qoo^{\sno}\, \qzz^{\sco}}{\mathsf P\left[ S_n = \sno \given  S_{n-1} = \sno \right]}
\end{align}
and $\mathsf P \left[ \, Y_n = 1 \given \sigma_{n-1} = (x_{n-1}, \sno ),\sigma_n=(x_n,s_n)\,\right] = 0$ otherwise. Similarly
\begin{align}
    &\mathsf P \left[ \, Y_n = 0 \given \sigma_{n-1} = (1, \sno ),\sigma_n=(0,\sno)\,\right] \\
    &\hspace*{2em}=\dfrac{ \qoo^{\sno}\, \qzz^{\sco}}{\mathsf P\left[ S_n = \sno \given  S_{n-1} = \sno \right]}
\end{align}
and $0$ otherwise.

In conclusion, the observation of a collision is the complementary event to those just described, and the corresponding probability, i.e. $\mathsf P \left[ \, Y_n = \collision \given \sigma_{n-1} = (x_{n-1}, \sno ),\sigma_n=(x_n,s_n)\,\right]$, can be derived accordingly.

\end{appendices}

\bibliographystyle{IEEEtran}
\bibliography{IEEEabrv,Template}

\end{document}